\documentclass[english]{lipics-v2021}
\newcommand{\PAPER}[1]{}
\newcommand{\LIPICS}[1]{#1}

\PAPER{
\usepackage{fullpage,amsthm,amssymb,graphicx}
\newtheorem{lemma}{Lemma}[section]
\newtheorem{theorem}[lemma]{Theorem}
\newtheorem{corollary}[lemma]{Corollary}
}

\LIPICS{
\usepackage[utf8]{inputenc}
\usepackage{amsthm,amssymb,graphicx}
   \hypersetup{%
      breaklinks,%
      colorlinks=true,%
      urlcolor=[rgb]{0.25,0.0,0.0},%
      linkcolor=[rgb]{0.5,0.0,0.0},%
      citecolor=[rgb]{0,0.2,0.445},%
      filecolor=[rgb]{0,0,0.4},
      anchorcolor=[rgb]={0.0,0.1,0.2}%
   }
}

\title{On the Fine-Grained Complexity of Small-Size Geometric Set Cover and Discrete $k$-Center for Small $k$}

\PAPER{
\author{Timothy M. Chan\thanks{Department of Computer Science,
University of Illinois at Urbana-Champaign (\{tmc,qizheng6,yyu51\}@illinois. edu).  Work supported in part by NSF Grant CCF-2224271. }
\and Qizheng He$^*$ \and Yuancheng Yu$^*$}
}

\LIPICS{

\author{Timothy M. Chan}{Department of Computer Science, University of Illinois at Urbana-Champaign, USA}{tmc@illinois.edu}{https://orcid.org/0000-0002-8093-0675}{Work supported by NSF Grant CCF-2224271.}

\author{Qizheng He}{Department of Computer Science, University of Illinois at Urbana-Champaign, USA}{qizheng6@illinois.edu}{https://orcid.org/0000-0002-2518-1114}{}

\author{Yuancheng Yu}{Department of Computer Science, University of Illinois at Urbana-Champaign, USA}{yyu51@illinois.edu}{}{}

\titlerunning{Small-Size Geometric Set Cover and 
Discrete $k$-Center for Small $k$}
\authorrunning{T.\,M. Chan, Q. He, and Y. Yu}

\Copyright{Timothy M. Chan, Qizheng He, and Yuancheng Yu}

\ccsdesc[100]{Theory of computation~Computational geometry}

\keywords{Geometric set cover, discrete $k$-center, conditional lower bounds}
}

\newtheorem{question}{Question}
\newcommand{\R}{\mathbb{R}}
\newcommand{\OO}{\widetilde{O}}

\newcommand{\IGNORE}[1]{}

\newcommand{\abs}[1]{\mathify{\left| #1 \right|}}
\newcommand{\mathify}[1]{\ifmmode{#1}\else\mbox{$#1$}\fi}
\newcommand{\ceil}[1]{\left\lceil #1 \right\rceil}

\newcommand{\norm}[1]{\left\| {#1} \right\|}

\newcommand{\timothy}[1]{\textcolor{red}{(Timothy: #1)}}
\newcommand{\qizheng}[1]{\textcolor{purple}{(Qizheng: #1)}}

\nolinenumbers
\hideLIPIcs

\begin{document}

\maketitle

\begin{abstract}
We study the time complexity of the discrete $k$-center problem and related (exact) geometric set cover problems when $k$ or the size of the cover is small.  We obtain a plethora of new results:
\begin{itemize}
\item We give the first subquadratic algorithm for \emph{rectilinear discrete 3-center} in 2D, running in $\OO(n^{3/2})$ time.
\item We prove a lower bound of $\Omega(n^{4/3-\delta})$ for rectilinear discrete 3-center in 4D, for any constant $\delta>0$, under a standard hypothesis about triangle detection in sparse graphs.
\item Given $n$ points and $n$ \emph{weighted} axis-aligned unit squares in 2D, we give the first subquadratic algorithm for finding a minimum-weight cover of the points by 3 unit squares, running in $\OO(n^{8/5})$ time.  We also prove a lower bound of $\Omega(n^{3/2-\delta})$ for the same problem in 2D, under the well-known APSP Hypothesis.  For arbitrary axis-aligned rectangles in 2D, our upper bound is $\OO(n^{7/4})$.
\item We prove a lower bound of $\Omega(n^{2-\delta})$ for Euclidean discrete 2-center in 13D, under the Hyperclique Hypothesis.  This lower bound nearly matches the straightforward upper bound of $\OO(n^\omega)$, if the matrix multiplication exponent $\omega$ is equal to 2.  
\item We similarly prove an $\Omega(n^{k-\delta})$ lower bound for Euclidean discrete $k$-center in $O(k)$ dimensions for any constant $k\ge 3$, under the Hyperclique Hypothesis.  This lower bound again nearly matches known upper bounds if $\omega=2$.
\item We also prove an $\Omega(n^{2-\delta})$ lower bound for the problem of finding 2 boxes to cover the largest number of points, given $n$ points and $n$ boxes in 12D\@.  This matches the straightforward near-quadratic upper bound.
\end{itemize}
\end{abstract}

\section{Introduction}

\subsection{The discrete $k$-center problem for small $k$}

The \emph{Euclidean $k$-center} problem is well-known in computational geometry and has a long history:
given a set $P$ of $n$ points in $\R^d$ and a number $k$, we want
to find $k$ congruent balls covering $S$, while minimizing the radius.
Euclidean 1-center can be solved in linear time for any constant dimension $d$
by standard techniques for low-dimensional linear programming or LP-type problems \cite{chazelle1996linear,clarkson1995vegas,dyer1986multidimensional,megiddo1983linear,welzl1991smallest}.  In a celebrated paper from SoCG'96,  Sharir~\cite{sharir1997near} 
gave the first $\OO(n)$-time\footnote{The $\OO$ notation hides polylogarithmic factors.} algorithm for Euclidean 2-center in $\R^2$, which represented a significant improvement over previous near-quadratic algorithms (the hidden logarithmic factors have since been reduced in a series of subsequent works~\cite{eppstein1996faster,chan1999more,Wang20,ChoiA21}).  
The problem is more difficult in higher dimensions:
the best time bound for Euclidean 2-center in $\R^d$ is about 
$n^d$ (see \cite{AgarwalES99,AgarwalAS13} for some results on the $\R^3$ case), and 
Cabello et al.~\cite{CabelloGKMR11} proved a conditional lower bound, ruling out $n^{o(d)}$-time algorithms, assuming the Exponential Time Hypothesis (ETH\@).
We are not aware of any work specifically addressing the Euclidean 3-center problem.

The $k$-center problem has also been studied under different metrics.  The most popular version after Euclidean is $L_\infty$ or \emph{rectilinear $k$-center}: here, we want to find
$k$ congruent hypercubes covering $P$, while minimizing the side length of the hypercubes.\footnote{All squares, rectangles, hypercubes, and boxes are axis-aligned in this paper.}  As expected, the rectilinear version is a little easier than the Euclidean.
Sharir and Welzl in SoCG'96~\cite{SharirW96} showed that rectilinear 3-center
problem in $\R^2$ can be solved in linear time, and that rectilinear 4-center and 5-center
in $\R^2$ can be solved in $\OO(n)$ time (the logarithmic factors have been subsequently improved by Nussbaum~\cite{Nussbaum97}).  Katz and Nielsen's work in SoCG'96~\cite{KatzN96}
implied near-linear-time algorithms for rectilinear 2-center in any constant dimension~$d$,
while Cabello et al.~in SODA'08~\cite{CabelloGKMR11}
gave an $O(n\log n)$-time algorithm for
rectilinear 3-center in any constant dimension~$d$.  Cabello et al.\ also proved a conditional lower bound for rectilinear 4-center, ruling out $n^{o(\sqrt{d})}$-time algorithms under ETH.

In this paper, we focus on a natural variant of the problem called \emph{discrete $k$-center}, which has also received considerable attention: here, given a set $P$ of $n$ points in $\R^d$ and a number $k$, we want
to find $k$ congruent balls covering $P$, while minimizing the radius, with the extra constraint that the centers of the chosen balls are from  $P$.\footnote{
Some authors define the problem slightly more generally, where the constraint is that the centers are from a second input set; in other words, the input consists of two sets of points (``demand points'' and ``supply points'').  The results of this paper will apply to both versions of the problem.}
The Euclidean discrete 1-center problem can be solved in $O(n\log n)$ time in $\R^2$
by a straightforward application of farthest-point Voronoi diagrams; it can also be solved in $O(n\log n)$ (randomized) time in $\R^3$ with more effort~\cite{Chan99}, and in subquadratic $\OO(n^{2-2/(\lceil d/2\rceil+1)})$ time for $d\ge 4$ by standard range searching techniques~\cite{AgaEriSURV,Matousek92}.
Agarwal, Sharir, and Welzl in SoCG'97~\cite{agarwal1998discrete} gave the 
first subquadratic algorithm for Euclidean discrete 2-center in $\R^2$,
running in $\OO(n^{4/3})$ time.

One may wonder whether Euclidean discrete 2-center in higher constant dimensions
could also be solved in subquadratic time via range searching techniques.
No results have been reported, but an $\OO(n^\omega)$-time algorithm
is not difficult to obtain, where $\omega<2.373$ denotes the matrix multiplication exponent:
by binary search, the problem reduces to finding two balls of a given radius $r$ with centers in $S$
covering $S$, which is equivalent to finding a pair $p,q\in S$ such that
$c_{pq}=\bigvee_{z\in S} (a_{pz}\wedge a_{zq})$ is false, where $a_{pz}$ is true iff $p$ and $z$ has distance more than $r$---this computation reduces to a Boolean matrix product.  
This approach works for arbitrary (not necessarily geometric) distance functions.  The  main question is whether
geometry could help in obtaining faster algorithms in the higher-dimensional  Euclidean setting, as Agarwal, Sharir, and Welzl were able to exploit successfully in $\R^2$:

\begin{question}\label{q1}
Is there an algorithm running in faster than $n^\omega$ time for the Euclidean discrete 2-center problem in higher constant dimensions?
\end{question}

We can similarly investigate the  
rectilinear version of the discrete $k$-center problem, which is potentially easier.
For example, the rectilinear discrete 2-center problem can be solved
in $\OO(n)$ time in any constant dimension~$d$, by a straightforward
application of orthogonal range searching, as reported in several papers \cite{BespamyatnikhK99,BespamyatnikhS99,KatzKS00}.
The approach does not seem to work for
the rectilinear discrete 3-center problem.
Naively, rectilinear discrete 3-center can be reduced to $n$ instances of
(some version of) rectilinear discrete 2-center, and solved
in $\OO(n^2)$ time.  However, no better results have been published, leading to the following questions:

\begin{question}\label{q2}
Is there a subquadratic-time algorithm for the rectilinear discrete 3-center problem?
\end{question}

\begin{question}\label{q3}
Are there lower bounds to show that the rectilinear discrete 3-center problem
does not have near-linear-time algorithm (and is thus strictly harder than
rectilinear discrete 2-center, or rectilinear continuous 3-center)?
\end{question}

Similar questions may be asked about rectilinear discrete $k$-center for $k\ge 4$.
Here, the complexity of the problem is upper-bounded by 
 $\OO(n^{\omega(\lfloor k/2\rfloor,1,\lceil k/2\rceil)})$,
where $\omega(a,b,c)$ denotes the exponent for multiplying an $n^a\times n^b$
and an $n^b\times n^c$ matrix: by binary search,
the problem reduces to finding $k$ hypercubes of a given edge length $r$ with centers in $S$
covering $S$, which is equivalent to finding a dominating set of size $k$
in the graph with vertex set $S$ where an edge $pz$ exists iff
the distance of $p$ and $z$ is more than $r$---the dominating set problem reduces to rectangular matrix multiplication
with the time bound stated, as observed by Eisenbrand and Grandoni~\cite{EisenbrandG04}.  Note that the difference $\omega(\lfloor k/2\rfloor,1,\lceil k/2\rceil)-k$  converges to 0 as $k\rightarrow\infty$ by known matrix multiplication bounds~\cite{Coppersmith82} (and is exactly 0 if $\omega=2$).

As $k$ gets larger compared to $d$,
a better upper bound of $n^{O(dk^{1-1/d})}$ is known for both the continuous and discrete $k$-center problem under the Euclidean and rectilinear metric~\cite{AgarwalP02,HwangCL93,HwangLC93}.  Recently, in SoCG'22, Chitnis and Saurabh~\cite{ChitnisS22} (extending earlier work by Marx~\cite{Marx05} in the $\R^2$ case) proved a nearly matching conditional lower bound for discrete $k$-center in $\R^d$, ruling out $n^{o(k^{1-1/d})}$-time algorithms under ETH\@.  However, these bounds do not answer our questions concerning very small $k$'s.
In contrast, the conditional lower bounds
by Cabello et al.~\cite{CabelloGKMR11} that we have mentioned earlier are about very small $k$ and so are more relevant,
but are only for the continuous version of the $k$-center problem.  (The continuous version behaves 
differently from the discrete version; see Tables~\ref{tbl1}--\ref{tbl2}.)

\begin{table}
\renewcommand{\arraystretch}{1.2}
\begin{tabular}{|l|ll|ll|ll|ll|}\hline
$k$ & \multicolumn{2}{c|}{Euclidean} & \multicolumn{2}{c|}{rectilinear} & \multicolumn{2}{c|}{Euclidean discrete} & \multicolumn{2}{c|}{rectilinear discrete}\\\hline
1 & $O(n)$ && $O(n)$ && $O(n\log n)$ && $O(n)$ &\\\hline
2 & $\OO(n)$ &\cite{sharir1997near} & $O(n)$ & \cite{SharirW96} & $\OO(n^{4/3})$
& \cite{agarwal1998discrete} & $\OO(n)$ &\\\hline
3 & 
&& $\OO(n)$&\cite{SharirW96} & 
&&
$\OO(n^{3/2})$ & {\bf (new)}\\\hline
\end{tabular}
\vspace{1ex}
\caption{Summary of results on $k$-center for small $k$ in $\R^2$.}\label{tbl1}
\end{table}

\begin{table}
\renewcommand{\arraystretch}{1.2}
\begin{tabular}{|l|ll|ll|ll|ll|}\hline
$k$ & \multicolumn{2}{c|}{Euclidean} & \multicolumn{2}{c|}{rectilinear} & \multicolumn{2}{c|}{Euclidean discrete} & \multicolumn{2}{c|}{rectilinear discrete}\\\hline
1 & $O(n)$ && $O(n)$ && $\OO(n^{2-2/(\lceil d/2\rceil +1)})\!\!\!\!\!\!\!\!$ && $O(n)$ &\\\hline
2 & $n^{O(d)}$ && $\OO(n)$ & \cite{KatzN96} & $\OO(n^\omega)$ && $\OO(n)$ & \\
  & CLB: $n^{\Omega(d)}\!\!\!\!\!$ &\cite{CabelloGKMR11} & && CLB: $\Omega(n^{2-\delta})\!\!\!\!\!\!$ & {\bf (new)}$\!$ 
  &&\\\hline
3 &  && $\OO(n)$ &\cite{CabelloGKMR11} & $\OO(n^{\omega(1,1,2)})$ && $\OO(n^2)$ &\\
  &  &&                      & & CLB: $\Omega(n^{3-\delta})\!\!\!\!\!\!$ & {\bf (new)}$\!$ 
  & CLB: $\Omega(n^{4/3-\delta})\!\!$  & {\bf (new)}$\!$ 
  \\\hline
4 &  &&  $n^{O(d)}$ && $\OO(n^{\omega(2,1,2)})$ && $\OO(n^3)$&\\
  &  && CLB: $n^{\Omega(\sqrt{d})}\!\!\!\!$ & \cite{CabelloGKMR11} & CLB: $\Omega(n^{4-\delta})\!\!\!\!\!\!$ & {\bf (new)}$\!$ 
  && \\\hline
\end{tabular}
\vspace{1ex}
\caption{Summary of results on $k$-center for small $k$ in $\R^d$ for
an arbitrary constant $d$.  (CLB stands for ``conditional lower bound''.)}\label{tbl2}
\end{table}

\subsection{The geometric set cover problem with small size $k$}

The decision version of the discrete $k$-center problem (deciding whether
the minimum radius is at most a given value) reduces to 
a \emph{geometric set cover} problem: given a set $P$ of $n$ points
and a set $R$ of $n$ objects, find the smallest subset of objects in $R$
that cover all points of $P$.
Geometric set cover has been extensively studied in the literature,
particularly from the perspective of approximation algorithms (since
for most types of geometric objects, set cover remains NP-hard); for example, see the references in~\cite{ChanH20}.
Here, we are interested in \emph{exact} algorithms for the case
when the optimal size $k$ is a small constant.

For the application to Euclidean/rectilinear $k$-center, the objects are congruent balls/ hypercubes, or by rescaling, unit balls/hypercubes, but other types of objects may be considered, such as arbitrary rectangles or boxes.  

We can also consider the \emph{weighted} version of the problem: here, given a set $P$ of $n$ points, a set $R$ of $n$ weighted objects, and a small constant $k$, we want to find a subset of $k$ objects in $R$
that cover all points of $P$, while minimizing the total weight of the chosen objects.

A ``dual'' problem is \emph{geometric hitting set}, which in the weighted case is the following: given a set $P$ of $n$ weighted points, a set $R$ of $n$ objects, and a small constant $k$, find a subset of $k$ points in $P$
that hit all objects of $R$, while minimizing the total weight of the chosen points.
(The continuous unweighted version, where 
the chosen points may be anywhere, is often called the \emph{piercing} problem.)  In the case of unit balls/hypercubes, hitting set is equivalent to set cover due to self-duality.

For rectangles in $\R^2$ or boxes in $\R^d$,
size-2 geometric set cover (unweighted or weighted) can be solved in $\OO(n)$ time, like discrete rectilinear 2-center~\cite{BespamyatnikhK99,BespamyatnikhS99,KatzKS00}, by orthogonal range searching.
Analogs to Questions \ref{q2}--\ref{q3} may be asked for size-3 geometric set cover for rectangles/boxes.

Surprisingly, the complexity of exact geometric set cover of small size $k$ has not received as much attention, but very recently in SODA'23, Chan~\cite{Chan23} initiated the study of similar questions for geometric independent set with small size $k$, for example, providing subquadratic algorithms and conditional lower bounds for size-4 independent set for boxes.

For larger $k$, hardness results by Marx and Pilipczuk~\cite{MarxP15} and
Bringmann et al.~\cite{BringmannKPL19}  ruled out 
$n^{o(k)}$-time algorithms  for size-$k$ geometric set cover for rectangles in $\R^2$ and unit hypercubes (or orthants) in $\R^4$, and 
$n^{o(\sqrt{k})}$-time algorithms for unit cubes (or orthants) in $\R^3$ under ETH\@.  But like the other fixed-parameter intractability results mentioned,
these proofs do not appear to imply any nontrivial lower bound for very small $k$ such as $k=3$.

\subsection{New results}

\subparagraph{New algorithms.}
In this paper, we answer Question~\ref{q2} in the affirmative for dimension $d=2$, by presenting the first subquadratic algorithms for rectilinear discrete 3-center in $\R^2$,
and more generally, for (unweighted and weighted) geometric size-3 set cover for unit squares, as well as arbitrary rectangles in $\R^2$.
More precisely, the time bounds of our algorithms are:

\begin{itemize}
\item $\OO(n^{3/2})$ for rectilinear discrete 3-center in $\R^2$ and unweighted size-3 set cover for unit squares in $\R^2$;
\item $\OO(n^{8/5})$ for weighted size-3 set cover for unit squares in $\R^2$; 
\item $\OO(n^{5/3})$ for unweighted size-3 set cover  for rectangles in $\R^2$;
\item $\OO(n^{7/4})$ for weighted size-3 set cover for rectangles in $\R^2$.
\end{itemize}

\subparagraph{New conditional lower bounds.}
We also prove the first nontrivial conditional lower bounds on the time complexity of rectilinear discrete 3-center and related size-3 geometric set cover problems.  More precisely, our lower bounds are:\footnote{Throughout this paper, $\delta$ denotes an arbitrarily small positive constant.}

\begin{itemize}
\item $\Omega(n^{3/2-\delta})$ for weighted size-3 set cover (or hitting set) for unit squares in $\R^2$, assuming the APSP Hypothesis;
\item $\Omega(n^{4/3-\delta})$ for rectilinear discrete 3-center in $\R^4$ and unweighted size-3 set cover (or hitting set) for unit hypercubes in $\R^4$, assuming the Sparse Triangle Hypothesis; 
\item $\Omega(n^{4/3-\delta})$ for unweighted size-3 set cover for boxes in $\R^3$, assuming the Sparse Triangle Hypothesis.
\end{itemize}

The lower bound in the first bullet is particularly attractive, since it implies that conditionally, our $\OO(n^{8/5})$-time algorithm for weighted size-3 set cover for unit squares in $\R^2$ is within a small factor (near $n^{0.1}$) from optimal, and that 
our $\OO(n^{7/4})$-time algorithm for weighted size-3 set cover for rectangles in $\R^2$ is within a factor near $n^{0.25}$ from optimal.
The second bullet answers Question~\ref{q3}, implying that rectilinear discrete 3-center is strictly harder than rectilinear discrete 2-center and
rectilinear (continuous) 3-center, at least when the dimension is 4 or higher.
(In contrast, 
rectilinear (continuous) 4-center is strictly harder than rectilinear discrete 4-center for sufficiently large constant dimensions~\cite{CabelloGKMR11}; see Table~\ref{tbl2}.)

In addition, we prove the following conditional lower bounds:

\begin{itemize}
\item $\Omega(n^{2-\delta})$ for Euclidean discrete 2-center in $\R^{13}$ and unweighted size-3 set cover (or hitting set) for unit balls in $\R^{13}$, assuming the Hyperclique Hypothesis;
\item $\Omega(n^{k-\delta})$ for Euclidean discrete $k$-center in $\R^{7k}$ and unweighted size-$k$ set cover for unit balls in $\R^{7k}$ for any constant $k\ge 3$, assuming the Hyperclique Hypothesis.
\end{itemize}

In particular, this answers Question~\ref{q1} in the negative if $\omega=2$ (as conjectured by some researchers): geometry doesn't help for Euclidean discrete 2-center when the dimension is a sufficiently large constant.  Similarly, the second bullet indicates that the upper bound $\OO(n^{\omega(\lfloor k/2\rfloor,1,\lceil k/2\rceil)})$
for Euclidean discrete $k$-center
is basically tight for any fixed $k\ge 3$ in a sufficiently large constant dimension, if $\omega=2$.
(See Tables \ref{tbl1}--\ref{tbl3}.)

\begin{table}
\renewcommand{\arraystretch}{1.2}
\begin{tabular}{|l|ll|ll|}\hline
objects & \multicolumn{2}{c|}{unweighted} & \multicolumn{2}{c|}{weighted}\\\hline
unit squares & $\OO(n^{3/2})$ & {\bf (new)} & $\OO(n^{8/5})$ & {\bf (new)}\\
 & & & CLB: $\Omega(n^{3/2-\delta})$ & {\bf (new)}\\\hline
rectangles  & $\OO(n^{5/3})$ & {\bf (new)} & $\OO(n^{7/4})$ & {\bf (new)}\\
 & & & CLB: $\Omega(n^{3/2-\delta})$ & {\bf (new)}\\\hline

\end{tabular}
\vspace{1ex}
\caption{Summary of results on size-3 geometric set cover in $\R^2$.}\label{tbl3}
\end{table}


Lastly, we prove a lower bound for a standard variant of set cover known as \emph{maximum coverage}:
given a set $P$ of $n$ points, a set $R$ of $n$ objects, and a small constant $k$, find $k$ objects in $R$ that cover the largest number (rather than all) of points of $P$.
Geometric versions of the maximum coverage problem have been studied before from the approximation algorithms perspective (e.g., see~\cite{BadanidiyuruKL12}).  It is also related to ``outliers'' variants of $k$-center problems
(where we allow a certain number of points to be uncovered), which have also been studied for small $k$
(e.g., see \cite{AgarwalP02}).
Recall that the size-2 geometric set cover problem for boxes in $\R^d$ can be solved in $\OO(n)$ time (which was why our attention was redirected to the size-3 case).  In contrast, we show that maximum coverage for boxes cannot be solved in near-linear time even for size $k=2$.  More precisely, we obtain the following lower bound:

\begin{itemize}
\item $\Omega(n^{2-\delta})$ for size-2 maximum coverage for unit hypercubes in $\R^{12}$, assuming the Hyperclique Hypothesis.
\end{itemize}

What is notable is that this lower bound is tight (up to $n^{o(1)}$ factors), regardless of $\omega$, since there is an obvious $\OO(n^2)$-time algorithm for boxes in $\R^d$ by answering $n^2$ orthogonal range counting queries---our result implies that this obvious algorithm can't be improved!

\subparagraph{On hypotheses from fine-grained complexity.}
Let us briefly state the hypotheses used.  

\begin{itemize}
\item The \emph{APSP Hypothesis} is among the three most popular hypotheses in fine-grained complexity~\cite{virgisurvey} (the other two being the 3SUM Hypothesis and the Strong Exponential Time Hypothesis): it asserts that there is no $O(n^{3-\delta})$-time algorithm for the all-pairs shortest paths problem for an arbitrary weighted graph with $n$ vertices (and $O(\log n)$-bit integer weights).  This hypothesis has been used extensively in the algorithms
literature (but less often in computational geometry).
\item The \emph{Sparse Triangle Hypothesis} asserts that there is no $O(m^{4/3-\delta})$-time algorithm for detecting a triangle (i.e., a 3-cycle) in a sparse unweighted graph with $m$ edges.  The current best upper bound for triangle detection, from a 3-decade-old paper by Alon, Yuster, and Zwick~\cite{AlonYZ97}, is
$\OO(m^{2\omega/(\omega+1)})$, which is $\OO(m^{4/3})$ if $\omega=2$.  (In fact, a stronger version of the hypothesis asserts that there is no
$O(m^{2\omega/(\omega+1)-\delta})$-time algorithm.)  As supporting evidence, it is known that certain ``listing'' or ``all-edges'' variants of the triangle detection problem 
have an $O(m^{4/3-\delta})$ lower bound, under the 3SUM Hypothesis or the APSP Hypothesis \cite{Patrascu10,WilliamsX20,ChanWX22}.  See \cite{AbboudBKZ22,JinXu22} for more discussion on the Sparse Triangle Hypothesis, and \cite{Chan23} for a recent application in computational geometry. 
\item The \emph{Hyperclique Hypothesis} asserts that there is no $O(n^{k-\delta})$-time algorithm for detecting a size-$k$ hyperclique
in an $\ell$-uniform hypergraph with $n$ vertices, for any fixed $k>\ell\ge 3$.
See \cite{LincolnWW18} for discussion on this hypothesis, and \cite{bringmann2022towards,Chan23,Kun22} for
some recent applications in computational geometry, including K\"{u}nnemann's breakthrough result on conditional lower bounds for Klee's measure problem~\cite{Kun22}.
\end{itemize}

\subparagraph{Techniques.}

Traditionally, in computational geometry, 
subquadratic algorithms with ``intermediate'' exponents between 1 and 2
tend to arise from the use of nonorthogonal range searching~\cite{AgaEriSURV} (Agarwal, Sharir, and Welzl's $\OO(n^{4/3})$-time algorithm for Euclidean discrete 2-center in $\R^2$~\cite{agarwal1998discrete} being one such example).  
Our subquadratic algorithms for rectilinear discrete 3-center in $\R^2$ and related
set-cover problems, which are about ``orthogonal'' or axis-aligned objects,
are different.
A natural first step is to use a $g\times g$ grid to divide into cases, for some carefully chosen parameter $g$.  Indeed, a grid-based approach was used in some recent subquadratic algorithms by Chan~\cite{Chan23} for size-4 independent set for boxes in 
any constant dimension, and size-5 independent set for rectangles in $\R^2$ 
(with running time $\OO(n^{3/2})$ and $\OO(n^{4/3})$ respectively).  However,
discrete 3-center or rectangle set cover is much more challenging than independent set
(for one thing, the 3 rectangles in the solution may intersect each other).  To make the grid approach work,
we need new original ideas (notably, a sophisticated argument to assign grid cells to rectangles,
which is tailored to the 2D case).  Still, the entire algorithm description fits in under 3 pages.

\IGNORE{
Our approach uses only orthogonal range searching, but
involves building a $g\times g$ grid, with a careful choice of parameter $g$ to
balance cost of several cases.  Recently, Chan~\cite{Chan23} applied similar strategies to obtain subquadratic algorithms for size-4 independent set for boxes in 
any constant dimension, and size-5 independent set for rectangles in $\R^2$ 
(with running time $\OO(n^{3/2})$ and $\OO(n^{4/3})$ respectively).  However,
discrete 3-center or rectangle set cover is more challenging than independent set
(for one thing, the 3 rectangles in the solution may intersect each other).
We need extra ideas (notably, an original argument to assign grid cells to rectangles),
which are specialized to the 2D case.
}

Our conditional lower bounds for rectilinear discrete 3-center and the corresponding set cover problem for unit hypercubes are proved by 
reduction from unweighted or weighted triangle finding in graphs.  It turns out
there is a simple reduction in $\R^2$ by exploiting weights.
However, lower bounds in the unweighted case (and thus the original rectilinear discrete 3-center problem) are much trickier.  We are able to design
a clever, simple reduction in $\R^6$ by hand, but reducing the dimension down
to 4 is far from obvious and we end up employing a \emph{computer-assisted search}, interestingly.  The final construction is still simple, and so is easy to verify by hand.

Our conditional lower bound proofs for Euclidean discrete 2-center, and more generally discrete $k$-center, are
inspired 
by a recent conditional hardness proof by Bringmann et al.~\cite{bringmann2022towards} from SoCG'22
on a different problem.  Specifically, they proved that deciding whether
the intersection graph of $n$ unit hypercubes in $\R^{12}$ has diameter~2 
requires near-quadratic time under the Hyperclique Hypothesis.
A priori, this diameter problem doesn't seem related to discrete $k$-center; moreover, it was a rectilinear problem, not Euclidean
(and we know that in contrast, rectilinear discrete 2-center has a near-linear upper bound!).  Our contribution is in realizing that Bringmann et al.'s approach is useful for Euclidean discrete 2-center and $k$-center, surprisingly.  To make the proof work though, we need some new technical ideas (in particular, an extra dimension for the $k=2$ case, and multiple extra dimensions for larger $k$, with carefully designed coordinate values).  Still, the final proof is not complicated to follow.  

Our conditional lower bound for size-2 maximum coverage for boxes is also proved using a similar technique, but again the adaptation is nontrivial, and we introduce some interesting counting arguments that proceed a bit differently from Bringmann et al.'s original proof for diameter (a problem that does not involve counting).

\section{Subquadratic Algorithms for 
Size-3 Set Cover for Rectangles in $\R^2$}
\label{sec:alg}

In this section, we describe the most basic version of our subquadratic algorithm to
solve the size-3 geometric set cover problem 
for weighted rectangles in $\R^2$.  
The running time is $\OO(n^{16/9})$.
Refinements of the algorithm
will be described in Sec.~\ref{app:improve}, where we will improve the time bound further to $\OO(n^{7/4})$, or even better for the unweighted case and unit square case.
The rectilinear discrete 3-center problem in $\R^2$
reduces to the unweighted unit square case by standard techniques~\cite{Chan99,FredericksonJ82} (see Sec.~\ref{app:improve}).

We begin with a lemma giving a useful geometric data structure:

\begin{lemma}\label{lem:pair}
For a set $P$ of $n$ points and a set $R$ of $n$ weighted rectangles in $\R^2$, 
we can build a data structure in $\OO(n)$ time and space, to support the following kind of queries:
given a pair of rectangles $r_1,r_2\in R$, we can find 
a minimum-weight rectangle $r_3\in R$ (if it exists) such that $P$ is covered by $r_1\cup r_2\cup r_3$,
in $\OO(1)$ time. 
\end{lemma}
\begin{proof}
By orthogonal range searching~\cite{AgaEriSURV,BerBOOK} on $P$, we can find the minimum/maximum $x$- and $y$-values among the points of $P$ in the complement of $r_1\cup r_2$  in $\OO(1)$ time (since the complement can be expressed as a union of $O(1)$ orthogonal ranges).  As a result, we obtain the minimum bounding box $b$ enclosing $P\setminus (r_1\cup r_2)$.  To finish, we find a minimum-weight rectangle in $R$ enclosing $b$; this is a ``rectangle enclosure'' query on $R$  and can be solved in $\OO(1)$ time, since it also reduces to orthogonal range searching (the rectangle $[x^-,x^+]\times [y^-,y^+]$ encloses the rectangle
$[\xi^-,\xi^+]\times [\eta^-,\eta^+]$  in $\R^2$ iff
the point $(x^-,x^+,y^-,y^+)$ 
lies in the box
$(-\infty,\xi^-]\times [\xi^+,\infty) \times (-\infty,\eta^-]\times [\eta^+,\infty)$
in $\R^4$).
\end{proof}

\begin{theorem}\label{thm:alg:basic}
Given a set $P$ of $n$ points and a set $R$ of $n$ weighted rectangles in $\R^2$, 
we can find $3$ rectangles $r_1^*,r_2^*,r_3^*\in R$ of minimum total weight (if they exist), such that
$P$ is covered by $r_1^*\cup r_2^*\cup r_3^*$, in $\OO(n^{16/9})$ time.
\end{theorem}
\begin{proof}
Let $B_0$ be the minimum bounding box enclosing $P$ (which touches 4 points).
If a rectangle of $R$ has an edge outside of $B_0$, we can eliminate that edge by extending the rectangle, making it unbounded.

%
Let $g$ be a parameter to be determined later.
Form a $g\times g$ (non-uniform) grid, where each column/row contains $O(n/g)$ 
rectangle vertices.

\begin{description}
\item[Step 1.]
For each pair of rectangles $r_1,r_2\in R$ that have vertical edges in a common column or horizontal edges in a common row, we query the data structure
in Lemma~\ref{lem:pair} to find a minimum-weight rectangle $r_3\in R$ (if exists) such that $P\subset r_1\cup r_2\cup r_3$, and add the triple $r_1r_2r_3$ to a list $L$.  The number of queried pairs $r_1r_2$ is $O(g\cdot (n/g)^2) = O(n^2/g)$,
and so this step takes $\OO(n^2/g)$ total time.

\LIPICS{\smallskip}
\item[Step 2.]
For each rectangle $r_1\in R$ and each of its horizontal (resp.\ vertical) edges $e_1$,
define $\gamma^-(e_1)$ and $\gamma^+(e_1)$ to be the leftmost and rightmost
(resp.\ bottommost and topmost) grid cell that intersects $e_1$ and contains a point of $P$ not covered by $r_1$.
We can naively find $\gamma^-(e_1)$ and $\gamma^+(e_1)$ by enumerating the $O(g)$ grid cells intersecting $e_1$ and
performing $O(g)$ orthogonal range queries; this takes $\OO(gn)$ total time.
For each rectangle $r_2\in R$ that has an edge intersecting $\gamma^-(e_1)$
or $\gamma^+(e_1)$, we query the data structure
in Lemma~\ref{lem:pair} to find a minimum-weight rectangle $r_3\in R$ (if exists) such that $P\subset r_1\cup r_2\cup r_3$, and add the triple $r_1r_2r_3$ to the list $L$.
The total number of queried pairs $r_1r_2$ is $O(n\cdot n/g) = O(n^2/g)$,
and so this step again takes $\OO(n^2/g)$ total time.
(This entire Step~2, and the definition of $\gamma^-(\cdot)$ and $\gamma^+(\cdot)$,
might appear mysterious at first, but their significance will be revealed later in Step~3.)

\LIPICS{\smallskip}
\item[Step 3.]
We guess the column containing each of the vertical edges of $r_1^*,r_2^*,r_3^*$ and
the row containing each of the horizontal edges of $r_1^*,r_2^*,r_3^*$;
there are at most 12 edges and so $O(g^{12})$ choices.  
Actually, 4 of the 12 edges are eliminated after extension, and so
the number of choices can be lowered to $O(g^8)$.

After guessing, we know which grid cells
are completely inside $r_1^*,r_2^*,r_3^*$ and which grid cells 
intersect which edges of $r_1^*,r_2^*,r_3^*$.  We may assume that
the vertical edges from different rectangles in $\{r_1^*,r_2^*,r_3^*\}$
are in different columns, and
the horizontal edges from different rectangles in $\{r_1^*,r_2^*,r_3^*\}$
are in different rows: if not, $r_1^*r_2^*r_3^*$ would have already
been found in Step~1.
In particular, we know combinatorially what the arrangement of $r_1^*,r_2^*,r_3^*$
looks like,
even though we do not know the precise coordinates and identities of $r_1^*,r_2^*,r_3^*$.

\begin{figure}
\centering\includegraphics[scale=1.05]{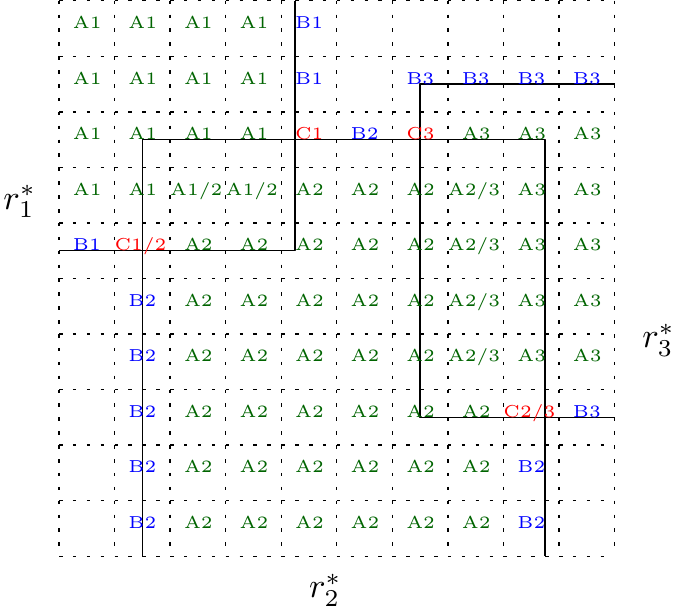}
\caption{Proof of Theorem~\ref{thm:alg:basic}: grid cells in Step 3.  The letter in a cell indicates its type (A, B, or C),
and the number (or numbers) in a cell indicates the index (or indices) $j\in\{1,2,3\}$
of the rectangle $r_j^*$ that the cell is assigned to.}
\label{fig:grid}
\end{figure}

We classify each grid cell $\gamma$ into the following types
(see Figure~\ref{fig:grid}):

\begin{itemize}
\item {\sc Type A}: $\gamma$ is completely contained in some $r_j^*\ (j\in\{1,2,3\})$. Here, we
\emph{assign} $\gamma$ to each such~$r_j^*$.
\item {\sc Type B}: $\gamma$ is not of type~A, and 
intersects an edge of exactly one rectangle $r_j^*$.
We \emph{assign} $\gamma$ to this $r_j^*$.
Observe that points in $P\cap\gamma$ can only be covered by $r_j^*$.
\item {\sc Type C}: $\gamma$ is not of type~A, and
intersects edges from two different rectangles in $\{r_1^*,r_2^*,r_3^*\}$.
W.l.o.g., suppose that $\gamma$ intersects a
horizontal edge $e_1^*$ of $r_1^*$ and a vertical edge $e_2^*$ of $r_2^*$;
note that the intersection point $v^*=e_1^*\cap e_2^*$ lies on the
boundary of the union $r_1^*\cup r_2^*\cup r_3^*$ (because $\gamma$ is not of type~A).
By examining the arrangement of $\{r_1^*,r_2^*,r_3^*\}$,
we know that at least one of the following is true: (i)~we can walk horizontally from $v^*$ to an endpoint of $e_1^*$ (or a point at infinity) while staying on the boundary of $r_1^*\cup r_2^*\cup r_3^*$, or (ii)~we can walk vertically from $v^*$ to an endpoint of  $e_2^*$ (or a point at infinity)  while staying on the boundary of $r_1^*\cup r_2^*\cup r_3^*$.

If (i) is true, we \emph{assign} $\gamma$ to $r_1^*$.
Observe that if there is a point in $P\cap\gamma$ not covered by $r_1^*$
(and if the guesses are correct),
then $\gamma$ must be equal to $\gamma^-(e_1^*)$ or $\gamma^+(e_1^*)$ (as defined in Step~2), and so $r_1^*r_2^*r_3^*$ would have already
been found in Step~2.
This is because except for $\gamma$, the grid cells
encountered while walking from $v^*$ to that endpoint of $e_1^*$ can intersect
only $r_1^*$ and so points in those cells can only be covered by $r_1^*$.

If (ii) is true, we \emph{assign} $\gamma$ to $r_2^*$ for a similar reason.
\end{itemize}

Note that there are at most $O(1)$ grid cells $\gamma$ of type~C; and the grid cells $\gamma$ of type~B form $O(1)$ contiguous blocks.  Let $\rho_j$ be the
union of all grid cells assigned to $r_j^*$.  Then $\rho_j$ is a rectilinear
polygon of $O(1)$ complexity.  We compute the minimum/maximum $x$- and $y$-values of
the points in $P\cap\rho_j$, by orthogonal range searching in $\OO(1)$ time.  As a result, we obtain the minimum bounding box $b_j$ enclosing $P\cap\rho_j$.
We find a minimum-weight rectangle $r_j\in R$ enclosing $b_j$, by
a rectangle enclosure query (reducible to orthogonal range searching, as before).  If $P\setminus (r_1\cup r_2\cup r_3)=\emptyset$ (testable by orthogonal range searching),
we add the triple $r_1r_2r_3$ (which should coincide with
$r_1^*r_2^*r_3^*$, if it has not been found earlier and if the guesses are correct) to $L$.  The total time over all guesses is $\OO(g^8)$.
\end{description}

At the end, we return a minimum-weight triple in $L$.  The overall running time
is $\OO(g^8 + n^2/g + gn)$.  Setting $g=n^{2/9}$ yields the theorem. 
\end{proof}

\section{Improved Subquadratic Algorithms for Size-3 Set Cover for Rectangles in $\R^2$ and Related Problems}\label{app:improve}

Continuing Section~\ref{sec:alg}, we describe refinements of our basic subquadratic algorithm to obtain improved time bounds in different cases.

\subsection{Improvement for unweighted rectangles}

\begin{lemma}\label{lem:maximal}
Given a set $S$ of points in a grid $\{1,\ldots,g\}^d$,
the number of \emph{maximal points} (i.e., points in $S$ that are not
strictly dominated by any other point in $S$) is $O(g^{d-1})$. 
\end{lemma}
\begin{proof}
Along each ``diagonal'' line of the form
$\{(t,c_1+t,\ldots,c_{d-1}+t): t\in\R\}$,
there can be at most one point of~$S$.  The grid can be covered by
$O(g^{d-1})$ such lines.
\end{proof}

\begin{theorem}\label{thm:alg:unwt}
The running time in Theorem~\ref{thm:alg:basic} can be improved
to $\OO(n^{5/3})$ if the rectangles in $R$ are unweighted.
\end{theorem}
\begin{proof}
In the unweighted case, it suffices to keep only the rectangles
that are \emph{maximal}, i.e., not strictly contained in other rectangles.
(We can test whether a rectangle is maximal in $\OO(1)$ time by a rectangle enclosure query,
reducible to orthogonal range searching.)  We run the algorithm
in the proof of Theorem~\ref{thm:alg:basic}.

For a $p$-sided rectangle $r_j^*$ (with $p\in\{1,\ldots,4\}$),
the number of choices
for the columns/rows containing the vertical/horizontal edges of $r_j^*$ is now
$O(g^{p-1})$ instead of $O(g^p)$ by Lemma~\ref{lem:maximal}, because the $p$-tuples of such column/row indices
correspond to maximal points of a subset of $\{1,\ldots,g\}^p$, by mapping $p$-sided rectangles to $p$-dimensional points.

This way, the number of guesses for the 3 rectangles $r_1^*,r_2^*,r_3^*$ can
be further lowered from $O(g^8)$ to $O(g^5)$,
and the overall running time
becomes $\OO(g^5 + n^2/g + gn)$.  Setting $g=n^{1/3}$ yields the theorem. 
\end{proof}

\subsection{Improvement for  weighted rectangles}

\begin{theorem}\label{thm:alg:improv}
The running time in Theorem~\ref{thm:alg:basic} can be improved
to $\OO(n^{7/4})$ for weighted rectangles.
\end{theorem}
\begin{proof}
We follow Steps 1--2 of the proof of Theorem~\ref{thm:alg:basic}
but will modify Step~3.

\begin{figure}
\centering\includegraphics[scale=1]{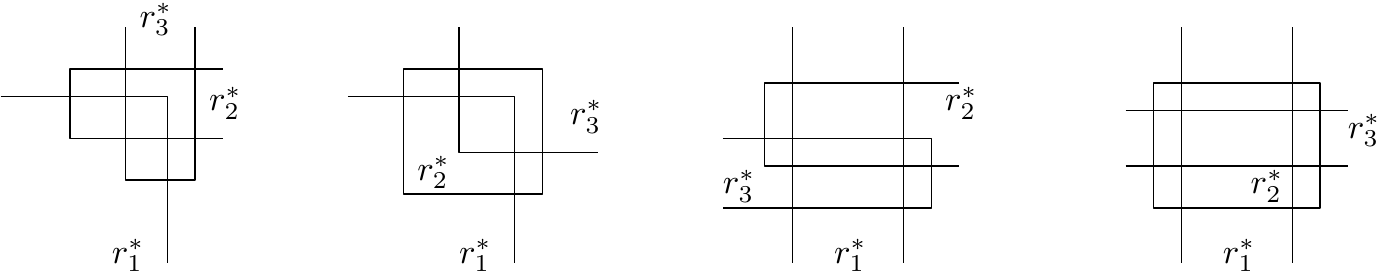}
\caption{Proof of Theorem~\ref{thm:alg:improv}: subcases of Case III, when there are no hidden edges.}
\label{fig:cases}
\end{figure}

\begin{itemize}
\item {\sc Case I:} Two of the rectangles in $\{r_1^*,r_2^*,r_3^*\}$
are disjoint.

W.lo.g., say that these two rectangles are $r_2^*$ and $r_3^*$, and
they are vertically separated, with $r_2^*$ to the left.
We guess the identity of $r_1^*$; there are $O(n)$ choices.
We guess a grid vertical line $\ell^*$ separating $r_1^*$ and $r_2^*$;
there are $O(g)$ choices.
(If a grid separating line does not exist, the answer would have
been found in Step~1.)
We find the minimum bounding box $b_2$
enclosing all points of $P\setminus r_1^*$ to the left of $\ell^*$ by orthogonal range searching, 
and then find a minimum-weight rectangle $r_2\in R$ enclosing $b_2$,
by a rectangle enclosure query.
We then query the data structure in Lemma~\ref{lem:pair} to find
a minimum-weight rectangle $r_3\in R$ (if exists) such that $P\subset
r_1^*\cup r_2\cup r_3$, and add the triple $r_1^*r_2r_3$ to the list $L$.
The total number of queries is $O(gn)$,
and so this case takes $\OO(gn)$ time.

\LIPICS{\smallskip}
\item {\sc Case II:} $r_1^*,r_2^*,r_3^*$ have at most 7 edges in total,
or have 8 edges but at least one of the edges is \emph{hidden}, i.e.,
does not appear in the boundary of the union $r_1^*\cup r_2^*\cup r_3^*$.

W.l.o.g., say that this hidden edge is $e_3^*$, defined by $r_3^*$.
We follow our earlier approach in Step~3, except that we need not
guess the column/row containing $e_3^*$; the number of guesses is thus
lowered from $O(g^8)$ to $O(g^7)$.  We know the 
grid cells of types A, B, and C that are
assigned to $r_1^*$ and $r_2^*$ (though not necessarily $r_3^*$),
and so we can compute the candidate rectangles $r_1$ and $r_2$ as before.
We then query the data structure in Lemma~\ref{lem:pair} to find
a minimum-weight rectangle $r_3\in R$ (if exists) such that $P\subset
r_1\cup r_2\cup r_3$, and add the triple $r_1r_2r_3$ to the list $L$.
The total number of queries is $O(g^7)$,
and so this case takes $\OO(g^7)$ time.

\LIPICS{\smallskip}
\item {\sc Case III:}
$r_1^*,r_2^*,r_3^*$ have a common intersection, and have
8 edges in total, none of which are hidden.

Some rectangle, say, $r_1^*$, must have at most 2 edges
(i.e., defined either by a vertical and a horizontal ray,
or by two parallel lines). 
By enumerating all scenarios,
we see that only 4 main subcases remain, as shown in Figure~\ref{fig:cases}
(all other subcases are symmetric).
In each of these subcases, one of the rectangles, say, $r_3^*$, 
has at most one edge touching 
$\partial r_1^* \setminus (r_2^*\cup r_3^*)$.
Denote this edge by $e_3^*$ (if exists).

We guess the identity of $r_2^*$; there are $O(n)$ choices.
Once $r_2^*$ is known, we add the vertical and horizontal lines through
the edges of $r_2^*$ to the grid.
We guess the row/column containing each of $r_1^*$'s horizontal/vertical edges and the row/column
containing $r_3^*$'s horizontal/vertical edge $e_3^*$ (if exists); there
are $O(g^3)$ such choices.
We then have enough information to determine the grid cells of types A, B, and C that are assigned to $r_1^*$ (though not necessarily $r_3^*$).
So we can compute the candidate rectangle $r_1$ as before.
We then query the data structure in Lemma~\ref{lem:pair} to find
a minimum-weight rectangle $r_3\in R$ (if exists) such that $P\subset
r_1\cup r_2^*\cup r_3$, and add the triple $r_1r_2^*r_3$ to the list $L$.
The total number of queries is $O(g^3 n)$,
and so this case takes $\OO(g^3n)$ time.
\end{itemize}

We do not know which case is true beforehand, but can try the algorithms for all cases
and return the best triple in $L$.
The overall running time
is $\OO(g^7 + g^3n + n^2/g)$.  Setting $g=n^{1/4}$ yields the theorem. 
\end{proof}

\subsection{Improvement for weighted unit squares}

\begin{theorem}\label{thm:alg:unit}
The running time in Theorem~\ref{thm:alg:basic} can be improved
to $\OO(n^{8/5})$ if the rectangles in $R$ are weighted unit squares.
\end{theorem}
\begin{proof}
We follow the approach in the proof of Theorem~\ref{thm:alg:basic},
but will modify Step~3.  W.l.o.g., assume that 
$r_1^*,r_2^*,r_3^*$ have centers ordered from left to right.
For unit squares, there are 3 possible cases (up to symmetry):

\begin{itemize}
\item {\sc Case I:} Two of the squares in $\{r_1^*,r_2^*,r_3^*\}$
are disjoint.

As in the proof of Theorem~\ref{thm:alg:improv}, this case
can be handled in $\OO(gn)$ time.

\begin{figure}
\centering\includegraphics[scale=0.7]{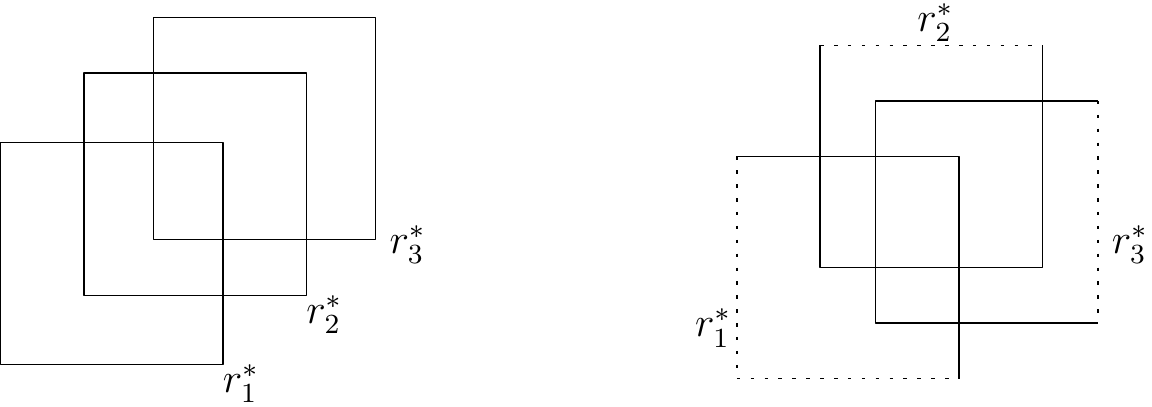}
\caption{Proof of Theorem~\ref{thm:alg:unit}: Case II (left) and
Case III (right).}
\label{fig:unit}
\end{figure}

\LIPICS{\smallskip}
\item {\sc Case II:} $r_1^*,r_2^*,r_3^*$ have a common intersection,
and $r_1^*$ has a lower center than $r_2^*$, which has a lower center than~$r_3^*$.  (See Figure~\ref{fig:unit} (left).)

We guess the identity of $r_2^*$; there are $O(n)$ choices.
Let $\hat{r}_2^*$ denote the region of all points below and left
of the top-right vertex of $r_2^*$.  Notice that $r_1^*\setminus r_2^*$ is contained
in $\hat{r}_2^*$, while $r_3^*\setminus r_2^*$ is outside $\hat{r}_2^*$.
We find the minimum bounding
box $b_1$ enclosing $P\cap \hat{r}_2^* \setminus r_2^*$, by orthogonal
range searching, and find a minimum-weight square $r_1\in R$ enclosing $b_1$ by a rectangle enclosure query.
We then query the data structure in Lemma~\ref{lem:pair} to find
a minimum-weight square $r_3\in R$ (if exists) such that $P\subset
r_1\cup r_2^*\cup r_3$, and add the triple $r_1r_2^*r_3$ to the list $L$.
The total number of queries is $O(n)$, and so this case takes $\OO(n)$ time.

\LIPICS{\smallskip}
\item {\sc Case III:}  $r_1^*,r_2^*,r_3^*$ have a common intersection,
and $r_1^*$ has a lower center than $r_3^*$, which has a lower center than~$r_2^*$.  (See Figure~\ref{fig:unit} (right).)

After extending the squares to unbounded rectangles, $r_1^*$ has 2 edges,
and $r_2^*$ and $r_3^*$ have 3 edges each.

We guess the column/row containing the vertical/horizontal edge of $r_1^*$,
and the columns containing the two vertical edges of $r_2^*$, and the rows containing the two horizontal edges of $r_3^*$;
there are $O(g^6)$ choices.
We need not guess the bottom horizontal edge of $r_2^*$ and
the left vertical edge of $r_3^*$, as these two edges are hidden.
We can assign grid cells of types B and C as before without knowing the hidden edges.  For type A,
each grid cell completely contained in $r_1^*$ can be assigned to $r_1^*$;
each grid cell that is strictly between the columns containing the two vertical
edges of $r_2^*$ and is on or above the row containing the top horizontal edge of $r_1^*$
can be assigned to $r_2^*$;
each grid cell that is strictly between the rows containing the two horizontal
edges of $r_3^*$ and is on or to the right of the column containing the right vertical edge of $r_1^*$
can be assigned to $r_3^*$.  The rest of the algorithm is the same.

Observe that the number of possible pairs of columns containing two
vertical edges of $r_2^*$ at unit distance is actually $O(g)$ instead of
$O(g^2)$ (to see this, imagine sweeping the grid by
two vertical lines at unit distance apart, creating $O(g)$ events).  Similarly, the number
of possible pairs of rows containing two
horizontal edges of $r_3^*$ at unit distance is $O(g)$.
With these observations, the number of guesses is further lowered to $O(g^4)$.
\end{itemize}

The overall running time is 
$\OO(g^4 + gn + n^2/g)$.  
Setting $g=n^{2/5}$ yields the theorem.
\end{proof}

\subsection{Improvement for unweighted unit squares and rectilinear discrete 3-center in $\R^2$}

\begin{theorem}\label{thm:alg:unit:unwt}
The running time in Theorem~\ref{thm:alg:basic} can be improved
to $\OO(n^{3/2})$ if the rectangles in $R$ are unweighted unit squares.
\end{theorem}
\begin{proof}
In the unweighted case, it suffices to keep only the rectangles
that are maximal (as in the proof of Theorem~\ref{thm:alg:unwt}).
We run the algorithm
in the proof of Theorem~\ref{thm:alg:unit}.

In Case III, the number of guesses for the column/row containing
the vertical/horizontal edge of $r_1^*$ is now $O(g)$ instead of $O(g^2)$
by Lemma~\ref{lem:maximal}.  The overall running time
becomes $\OO(g^3 + gn + n^2/g)$.   
Setting $g=\sqrt{n}$ yields the theorem.
\end{proof}

\begin{corollary}
Given a set $P$ of $n$ points in $\R^2$, we
can solve the rectilinear discrete 3-center problem
in $\OO(n^{3/2})$ time.
\end{corollary}
\begin{proof}
The decision problem---deciding whether the optimal radius is at most
a given value $r$---amounts to finding 3 squares covering $P$,
among $n$ squares with centers at $P$ and side length $2r$.
The decision problem can thus be solved in $\OO(n^{3/2})$ time
by Theorem~\ref{thm:alg:unwt} (after rescaling to make the side length unit).
We can then solve the original optimization problem
by Frederickson and Johnson's technique~\cite{FredericksonJ82}, with an
extra logarithmic factor (or alternatively by Chan's randomized technique~\cite{Chan99}, without the logarithmic factor).
\end{proof}

\section{Conditional Lower Bounds for 
Size-3 Set Cover for Boxes}
\label{sec:lb:rectilinear}

In this section, we prove conditional lower bounds for 
size-3 set cover for boxes in certain dimensions (rectilinear discrete 3-center is related to size-3 set cover for unit hypercubes).  We begin with the weighted version, which is more straightforward and has a simple proof, and serves as a good warm-up to the more challenging, unweighted version later.





\subsection{Weighted size-3 set cover for unit squares in $\R^2$} 

An \emph{orthant} (also called a dominance range)
refers to a $d$-sided box in $\R^d$ which is unbounded along each of the $d$ dimensions.  (Note that orthants may be oriented in $2^d$ ways.)  To obtain a lower bound for the unit square  or unit hypercube case, it suffices to obtain a lower bound for the orthant case, since
we can just replace each orthant with a hypercube with a sufficiently large side length $M$, and then rescale by a $1/M$
factor.

\begin{theorem}
Given a set $P$ of $n$ points and a set $R$ of $n$ weighted orthants in $\R^2$, finding 3 orthants in $R$ of minimum total weight that cover $P$
requires $\Omega(n^{3/2-\delta})$ time for any constant $\delta>0$, assuming the APSP Hypothesis.
\end{theorem}
\begin{proof}
The APSP Hypothesis is known to be equivalent~\cite{WilliamsW18} to the hypothesis that finding
a minimum-weight triangle in a weighted graph with $n$ vertices requires $\Omega(n^{3-\delta})$ time for any constant $\delta>0$.
We will reduce the minimum-weight triangle problem on a graph with $n$ vertices and $m$ edges ($m\in [n,n^2]$) to the weighted size-3 set cover problem for $O(m)$ points and orthants in $\R^2$.
Thus, if there is an $O(m^{3/2-\delta})$-time algorithm for the latter problem, there would be
an algorithm for the former problem with running time $O(m^{3/2-\delta})\le O(n^{3-2\delta})$, refuting the hypothesis.

Let $G=(V,E)$ be the given weighted graph with $n$ vertices and $m$ edges.
Without loss of generality, assume that all edge weights are in $[0,0.1]$, and that
$V\subset [0,0.1]$, i.e., vertices are labelled by numbers that are rescaled to lie in $[0,0.1]$.
Assume that $0\in V$ and $0.1\in V$.


\begin{figure}\centering
    \includegraphics[width=.3\textwidth]{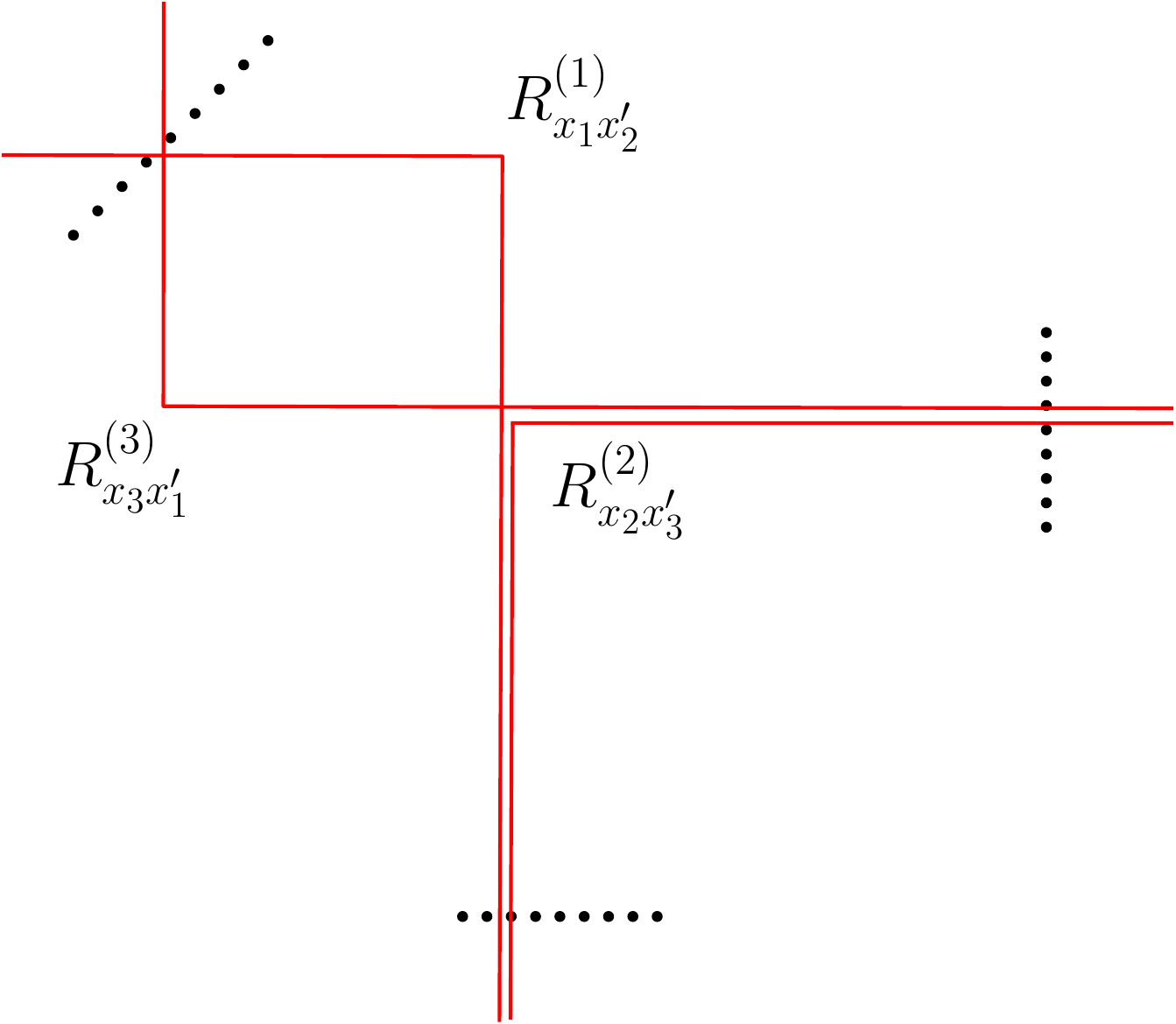}\\
    \caption{Reduction from the minimum-weight triangle problem to weighted size-3 set cover for orthants in $\R^2$.} 
    \label{fig:3_rects_weighted_LB}
\end{figure}

\subparagraph{The reduction.} For each vertex $t\in V$, 
we create three points $(t,1+t)$, $(2,t)$, and $(1+t,-1)$ (call them of type 1, 2, and 3, respectively).

Create the following orthants in $\R^2$:
\medskip
\[ 
\begin{array}{rccccr}
\forall x_1x_2'\in E:\ R_{x_1x_2'}^{(1)}&=& (-\infty,1+x_2') &\times& (-\infty,1+x_1] & \qquad\qquad\mbox{(type 1)}\\
\forall x_2x_3'\in E:\ R_{x_2x_3'}^{(2)}&=& [1+x_2,\infty) &\times& (-\infty,x_3') & \mbox{(type 2)}\\ 
\forall x_3x_1'\in E:\ R_{x_3x_1'}^{(3)}&=& (x_1',\infty) &\times& [x_3,\infty) & \mbox{(type 3)}
\end{array}
\]
\IGNORE{
\begin{align*}
R_{x_1x_2'}^{(1)}&=\ (-\infty,1+x_2'] \times (-\infty,1+x_1),\\
R_{x_2x_3'}^{(2)}&=\ (1+x_2,\infty) \times (-\infty,x_3],\\ 
R_{x_3x_1'}^{(3)}&=\ [x_1',\infty) \times (x_3,\infty).
\end{align*}
}%
The weight of each orthant is set to be the number of points it covers plus 
the weight of the edge it represents.
The total number of points and orthants is $O(n)$ and $O(m)$ respectively.
The reduction is illustrated in Figure~\ref{fig:3_rects_weighted_LB}.

\subparagraph{Correctness.}
We prove that the minimum-weight triangle in $G$ has weight $w$ (where $w\in [0,0.3]$) iff the optimal weighted size-3 set cover has weight $3n+w$.

Any feasible solution (if exists) must use an orthant of each type, since the point $(0,1)$ of type $1$ 
(resp.\ the point $(2,0.1)$ of type $2$, and the point $(1.1,-1)$ of type $3$) can only be covered by an orthant of type 1 (resp.\ 3 and 2). 
So, the three orthants in the optimal solution must be of the form $R^{(1)}_{x_1x_2'}$, $R^{(2)}_{x_2x_3'}$ and $R^{(3)}_{x_3x_1'}$ for some $x_1x_2',x_2x_3',x_3x_1'\in E$.

If $x_1<x_1'$, some point (of type 1) would be uncovered; on the other hand,
if $x_1>x_1'$, some point (of type 1) would be covered twice, and the total weight would then be at least $3n+1$.  
Thus, $x_1=x_1'$. 
 Similarly, $x_2=x_2'$ and $x_3=x_3'$.  So, $x_1x_2x_3$ forms a triangle in~$G$.
We conclude that the minimum-weight solution $R^{(1)}_{x_1x_2}$, $R^{(2)}_{x_2x_3}$ and $R^{(3)}_{x_3x_1}$ correspond to the minimum-weight triangle $x_1x_2x_3$ in $G$.
\IGNORE{
The next step is to prove these three rectangles correspond to the three edges of a triangle. First consider all points of type 1 with the form $(t,1+t)$. 
 Observe that $R^{(2)}_{x_2x_3'}$ cannot cover any of them, $R^{(1)}_{x_1x_2'}$ can cover all points corresponding to $t<x_1$, and $R^{(3)}_{x_3x_1'}$ can cover all points corresponding to $t\geq x_1'$. So in order to cover all of them, $x_1\geq x_1'$ must hold. Also if $x_1>x_1'$, at least one point of type 1 is covered twice, thus the total weight of the solution exceeds $3n$. Therefore $x_1=x_1'$ must hold. In other words, the points of type 1 are used to encode the equality $x_1=x_1'$. Similarly, $x_2=x_2'$ and $x_3=x_3'$ also hold, so there is a triangle $x_1x_2x_3$. For the other direction, it is easy to see having a triangle $x_1x_2x_3$ with weight $w$ ensures a weighted size-3 set cover $R^{(1)}_{x_1x_2}$, $R^{(2)}_{x_2x_3}$ and $R^{(3)}_{x_3x_1}$ with weight $3n+w$, which is optimal.
}
%
\end{proof}



\subsection{Unweighted size-3 set cover for boxes in $\mathbb{R}^3$}\label{sec:unweighted_boxes_3D}

Our preceding reduction 
uses weights to ensure equalities of two variables representing vertices. For the unweighted case, this does not work.  We propose a different way to force equalities, by using an extra dimension and extra sides (i.e., using boxes instead of orthants), with some carefully chosen coordinate values.


\begin{theorem}\label{thm:lb:box:3d}
Given a set $P$ of $n$ points and a set $R$ of $n$ unweighted axis-aligned boxes in $\mathbb{R}^3$, deciding whether there exist 3 boxes in $R$ that cover $P$ requires $\Omega(n^{4/3-\delta})$ time for any constant $\delta>0$, assuming the Sparse Triangle Hypothesis.
\end{theorem}
\begin{proof}
We will reduce the triangle detection problem on a graph with $m$ edges to the unweighted size-3 set cover problem for $O(m)$ points and boxes in $\R^3$.
Thus, if there is an $O(m^{4/3-\delta})$-time algorithm for the latter problem, there would be
an algorithm for the former problem with running time $O(m^{4/3-\delta})$, refuting the hypothesis.

Let $G=(V,E)$ be the given unweighted sparse graph with $n$ vertices and $m$ edges ($n\le m$). Without loss of generality, assume that $V\subset [0,0.1]$,
and $0\in V$ and $0.1\in V$.

\subparagraph{The reduction.}
For each vertex $t\in V$, create six points 
\[ 
\begin{array}{rrrr}
(-1+t,&0,&2+t) &\qquad\qquad\mbox{(type 1)}\\
(1+t,&0,&-2+t) &\mbox{(type 2)}\\
(2+t,&-1+t,&0) &\mbox{(type 3)}\\
(-2+t,&1+t,&0) &\mbox{(type 4)}\\
(0,&2+t,&-1+t) &\mbox{(type 5)}\\
(0,&-2+t,&1+t) &\mbox{(type 6)}
\end{array}\]

Create the following boxes in $\mathbb{R}^3$:
\IGNORE{
\[ 
\begin{array}{rccccccr}
\forall x_1x_2'&\in &E, \text{ let}\hspace{4em}\ &\\
R_{x_1x_2'}^{(1)}&=&(-1+x_1,1+x_1)&\times& [-2+x_2',2+x_2']&\times& \mathbb{R} &\ \ \ \mbox{(type 1)}\\
\forall x_2x_3'&\in &E, \text{ let}\hspace{4em}\ &\\
R_{x_2x_3'}^{(2)}&=&[-2+x_3',2+x_3']&\times& \mathbb{R}&\times& (-1+x_2,1+x_2) &\mbox{(type 2)}\\
\forall x_3x_1'&\in &E, \text{ let}\hspace{4em}\ &\\
R_{x_3x_1'}^{(3)}&=&\mathbb{R}&\times& (-1+x_3,1+x_3)&\times& [-2+x_1',2+x_1'] &\mbox{(type 3)}
\end{array}\]
}
\[ 
\begin{array}{rccccccr}
\forall x_1x_2'\in E:\ R_{x_1x_2'}^{(1)}&=&(-1+x_1,1+x_1)&\times& [-2+x_2',2+x_2']&\times& \mathbb{R} &\ \ \ \\
\forall x_2x_3'\in E:\ R_{x_2x_3'}^{(2)}&=&[-2+x_3',2+x_3']&\times& \mathbb{R}&\times& (-1+x_2,1+x_2)&\\
\forall x_3x_1'\in E:\ R_{x_3x_1'}^{(3)}&=&\mathbb{R}&\times& (-1+x_3,1+x_3)&\times& [-2+x_1',2+x_1']&
\end{array}\]
(call them of type 1, 2, and 3, respectively).

\IGNORE{
\begin{align*}
R_{x_1x_2'}&=(-1+x_1,1+x_1)\times [-2+x_2',2+x_2']\times \mathbb{R},\\
R_{x_2x_3'}&=[-2+x_3',2+x_3']\times \mathbb{R}\times (-1+x_2,1+x_2),\\
R_{x_3x_1'}&=\mathbb{R}\times (-1+x_3,1+x_3)\times [-2+x_1',2+x_1'].
\end{align*}
}

\subparagraph{Correctness.}
We prove that a size-3 set cover exists iff a triangle exists in $G$.  

Any feasible solution (if exists)
must use a box of each type, since the point $(-1,0,2)$ of type $1$ 
(resp.\ the point $(2,-1,0)$ of type $3$, and the point $(0,2,-1)$ of type $5$) can only be covered by a box of type 3 (resp.\ 2 and 1). So, the three boxes in a feasible solution must be of the form $R_{x_1x_2'}^{(1)}$, $R_{x_2x_3'}^{(2)}$ and $R_{x_3x_1'}^{(3)}$ for some $x_1x_2',x_2x_3',x_3x_1'\in E$.

Consider points of type 1 with the form $(-1+t,0,2+t)$. The box $R_{x_2x_3'}^{(2)}$ cannot cover any of them due to the third dimension.  The box $R_{x_1x_2'}^{(1)}$  covers all such points corresponding to $t>x_1$, and the box $R_{x_3x_1'}^{(3)}$ covers all such points corresponding to $t\leq x_1'$. So, all points of type 1 are covered iff 
$x_1\leq x_1'$.
Similarly, all points of type 2 are covered iff $x_1'\le x_1$.  Thus, all points of type 1--2 are covered iff $x_1=x_1'$.
By a symmetric argument, all points of type 3--4 are covered iff $x_3=x_3'$; and all points of type 5--6 are covered iff $x_2=x_2'$.  We conclude that a feasible solution exists iff a triangle $x_1x_2x_3$ exists in $G$.
%
%
\end{proof}  

We remark that the boxes above can be made fat, with side lengths between 1 and a constant (by replacing $\R$ with an interval of a sufficiently large constant length).

\subsection{Unweighted size-3 set cover for unit hypercubes in $\R^4$}

Our preceding lower bound for unweighted size-3 set cover for boxes in $\R^3$ immediately implies a lower bound for orthants (and thus unit hypercubes) 
in $\R^6$, since
the point $(x,y,z)$ is covered by the box
$[a^-,a^+]\times [b^-,b^+]\times [c^-,c^+]$ in $\R^3$
iff the point $(x,x,y,y,z,z)$ is covered by
the orthant
$[a^-,\infty)\times (-\infty,a^+] \times
[b^-,\infty)\times (-\infty,b^+] \times
[c^-,\infty)\times (-\infty,c^+]$ in $\R^6$.


A question remains: can the dimension 6 be lowered?  Intuitively, there seems to be some wastage in the above construction: there are several 0's in the coordinates of the points, and several $\R$'s in the definition of the boxes, and these get doubled after the transformation to 6 dimensions.  However, it isn't clear how to rearrange coordinates to eliminate this wastage: we would have to give up this nice symmetry of our construction, and there are too many combinations to try.  We ended up writing a computer program to exhaustively try all these different combinations, and eventually find a construction that lowers the dimension to 4!  Once it is found, correctness is straightforward to check, as one can see in the proof below.


\begin{theorem}\label{thm:lb:orthant:4d}
Given a set $P$ of $n$ points and a set $R$ of $n$ unweighted orthants in $\mathbb{R}^4$, deciding whether there exists a size-3 set cover requires $\Omega(n^{4/3-\delta})$ time for any constant $\delta>0$, assuming the Sparse Triangle Hypothesis.
\end{theorem}
\begin{proof}
We will reduce the triangle detection problem on a graph with $m$ edges to the unweighted size-3 set cover problem for $O(m)$ points and orthants in $\R^4$.

Let $G=(V,E)$ be the given unweighted graph with $n$ vertices and $m$ edges
($n\le m$). Without loss of generality, assume that $V\subset [0,0.1]$,
and $0\in V$ and $0.1\in V$.

\subparagraph{The reduction.}
For each vertex $t\in V$, create six points
\[ 
\begin{array}{rrrrr}
(0+t,&2+t,&-0.5,&-0.5) &\qquad\qquad\mbox{(type 1)}\\
(2-t,&0-t,&-0.5,&-0.5) &\mbox{(type 2)}\\
(1-t,&0.5,&1+t,&1.5) &\mbox{(type 3)}\\
(0.5,&1+t,&0.5,&2-t) &\mbox{(type 4)}\\
(-0.5,&-0.5,&2-t,&0-t) &\mbox{(type 5)}\\
(-0.5,&-0.5,&0+t,&1+t) &\mbox{(type 6)}
\end{array}\]
\IGNORE{
\begin{align*}
(0+t,2+t,-0.5,-0.5),\\
(2-t,0-t,-0.5,-0.5),\\
(1-t,0.5,1+t,1.5),\\
(0.5,1+t,0.5,2-t),\\
(-0.5,-0.5,2-t,0-t),\\
(-0.5,-0.5,0+t,1+t).
\end{align*}
}

Create the following orthants in $\mathbb{R}^4$:
\IGNORE{
\[ 
\begin{array}{rcccccccc}
\forall x_1x_2'&\in &E, \text{ let}\hspace{3em}\ &\\
R_{x_1x_2'}^{(1)}&=&[0+x_1,+\infty)&\times& [0-x_1,+\infty)&\times& (-\infty,1+x_2')&\times& (-\infty,2-x_2')\\
\forall x_2x_3'&\in &E, \text{ let}\hspace{3em}\ &\\
R_{x_2x_3'}^{(2)}&=&(-\infty,1-x_2]&\times& (-\infty,1+x_2]&\times& (0+x_3',+\infty)&\times& (0-x_3 ',+\infty)\\
\forall x_3x_1'&\in &E, \text{ let}\hspace{3em}\ &\\
R_{x_3x_1'}^{(3)}&=&(-\infty,2-x_1')&\times& (-\infty,2+x_1')&\times& (-\infty,2-x_3]&\times& (-\infty,1+x_3]
\end{array}
\]
}
\[ 
\begin{array}{rcccccccc}
\forall x_1x_2' \in E:\ 
R_{x_1x_2'}^{(1)}\!\!\!\!&=&\!\![0+x_1,+\infty)\!\!\!\!&\times&\!\! [0-x_1,+\infty)\!\!\!\!&\times&\!\! (-\infty,1+x_2')\!\!\!\!&\times&\!\! (-\infty,2-x_2')\\
\forall x_2x_3' \in E:\ 
R_{x_2x_3'}^{(2)}\!\!\!\!&=&\!\! (-\infty,1-x_2]\!\!\!\!&\times&\!\! (-\infty,1+x_2]\!\!\!\!&\times&\!\! (0+x_3',+\infty)\!\!\!\!&\times&\!\! (0-x_3 ',+\infty)\\
\forall x_3x_1'\in E:\
R_{x_3x_1'}^{(3)}\!\!\!\!&=&\!\! (-\infty,2-x_1')\!\!\!\!&\times&\!\! (-\infty,2+x_1')\!\!\!\!&\times&\!\! (-\infty,2-x_3]\!\!\!\!&\times&\!\! (-\infty,1+x_3]
\end{array}
\]
(call them of type 1, 2, and 3, respectively).

\subparagraph{Correctness.}
We prove that a size-3 set cover exists iff a triangle exists in $G$.  

Any feasible solution (if exists)
must use an orthant of each type, as one can easily check (like before).
 So, the three orthants in a feasible solution must be of the form $R_{x_1x_2'}^{(1)}$, $R_{x_2x_3'}^{(2)}$ and $R_{x_3x_1'}^{(3)}$ for some $x_1x_2',x_2x_3',x_3x_1'\in E$.
 
Consider points of type 1 with the form $(0+t,2+t,-0.5,-0.5)$.
The orthant $R_{x_2x_3'}^{(2)}$  cannot cover any of them due to the third dimension.
The orthant $R_{x_1x_2'}^{(1)}$ covers all such points corresponding to $t\ge x_1$, and
the orthant $R_{x_3x_1'}^{(3)}$ covers all such points corresponding to $t< x_1'$.
So, all points of type 1 are covered iff $x_1\le x_1'$.
By similar arguments, it can be checked that
all points of type 2 are covered iff $x_1\ge x_1'$;
all points of type 3 are covered iff $x_2\le x_2'$;
all points of type 4 are covered iff $x_2\ge x_2'$;
all points of type 5 are covered iff $x_3\le x_3'$;
all points of type 6 are covered iff $x_3\ge x_3'$.
We conclude that a feasible solution exists iff a triangle $x_1x_2x_3$ exists in $G$.
\end{proof}

\IGNORE{
\begin{align*}
R_{x_1x_2'}&=[0+x_1,+\infty)\times [0-x_1,+\infty)\times (-\infty,1+x_2']\times (-\infty,2-x_2'],\\
R_{x_2x_3'}&=(-\infty,1-x_2]\times (-\infty,1+x_2]\times [0+x_3',+\infty)\times [0-x_3 ',+\infty),\\
R_{x_2x_3'}&=(-\infty,2-x_3']\times (-\infty,2+x_3']\times (-\infty,2-x_2]\times (-\infty,1+x_2].
\end{align*}
}





\IGNORE{
As mentioned, a lower bound for orthants in $\R^4$ automatically implies a lower bound for unit hypercubes in $\R^4$.  In our construction, we can replace each orthant with a hypercube of side length 10 (i.e., $L_\infty$-radius 5), keeping the corner vertex the same, and then rescale.
}

In the computer search, we basically tried different choices of points with coordinate values of the form $c\pm t$ or $c$ for some constant $c$, and orthants defined by intervals of the form $[c\pm x_j,+\infty)$ or $(-\infty, c\pm x_j]$ (closed or open) for some variable $x_j$ (or $x_j'$).  The constraints are not exactly easy to write down, but are self-evident as we simulate the correctness proof above.  Naively, the number of cases is in the order of $10^{14}$, but can be drastically
reduced to about $10^7$ with some optimization and careful pruning of the search space.
The C++ code is not long (under 150 lines) and, after incorporating pruning, runs in under a second.


\IGNORE{
\qizheng{(finished) In the previous comment, "searching the whole space" means searching in the optimized space that exploits symmetry and uses pruning. Without optimization it should take days. Let me count the total number of cases: without optimization, in dimension 4, there are \[\underset{\text{\# orderings for the constants}}{(3!)^4}\cdot \underset{\text{\# possible directions for the intervals}}{(2^3)^4}\cdot \underset{\text{\# possible directions for }+/-x_i}{(2^3)^4}\]
\[\cdot \underset{\text{\# perfect matchings on 12 vertices, to pair the two occurrences for each }x_i}{10395}=226021311774720\] configurations (which is more than days; before when I estimate "days" I probably already used some symmetry).\\
After optimization, there are only \[\underset{\text{\# ordering and directions for the constants and variables}}{12^4}\cdot \underset{\text{\# matchings for pairing the two occurrences for each }x_i}{2784}=57729024\] cases.}
}

It is now straightforward to modify the above lower bound proof for unweighted orthants (or unit hypercubes) in $\R^4$ to the 
rectilinear discrete 3-center problem in $\R^4$; see Sec.~\ref{app:d3c}.
In Sec.~\ref{sec:k_clique}, we also prove a higher conditional lower bound for 
weighted size-6 set cover for rectangles in $\mathbb{R}^2$.

\section{Other Conditional Lower Bounds for Small-Size Set Cover for Boxes and Related Problems}

Continuing Section~\ref{sec:lb:rectilinear}, we prove a few more conditional lower bounds for related problems.

\subsection{Rectilinear discrete 3-center in $\R^4$}\label{app:d3c}

It is easy to modify our conditional lower bound proof for size-3 set cover for orthants in $\R^4$ 
(Theorem~\ref{thm:lb:orthant:4d})
to obtain a lower bound for the rectilinear discrete 3-center problem in $\R^4$.

First recall that a lower bound for orthants in $\R^4$ automatically implies a lower bound for unit hypercubes in $\R^4$.  In our construction, we can replace each orthant with a hypercube of side length 10 (i.e., $L_\infty$-radius 5), keeping the corner vertex the same, and then rescale.

For rectilinear discrete 3-center in $\R^4$, the new point set consists of the constructed points and the centers of the constructed hypercubes 
of side length 10 (before rescaling), together with
three auxiliary points: $(9.5,9.5,-8.5,-7.5),(-8.5,-8.5,9.5,9.5),
(-7.5,-7.5,-7.5,-8.5)$.  The purpose of the auxiliary points is to force the 3 centers in the solution to be from the centers of the constructed hypercubes.


\begin{corollary}\label{cor:lb:d3c}
For any constant $\delta > 0$, there is no $O(n^{4/3-\delta})$-time algorithm for rectilinear discrete 3-center in $\mathbb{R}^{4}$, assuming the Sparse Triangle Hypothesis.
\end{corollary}


\subsection{Weighted size-6 set cover for rectangles in $\mathbb{R}^2$}\label{sec:k_clique}

In Section~\ref{sec:lb:rectilinear}, we have obtained a superlinear
conditional lower bound for weighted size-3 set cover for rectangles in $\R^2$.  Another direction is to investigate the smallest $k$ for which we can obtain a near-quadratic lower bound for weighted size-$k$ set cover.
We prove such a result for $k=6$.

For a fixed $\ell\ge 3$, the \emph{Weighted $\ell$-Clique Hypothesis}~\cite{virgisurvey} asserts that 
there is no $O(n^{\ell-\delta})$ time algorithm for finding an $\ell$-clique of minimum total edge weight in an $n$-vertex graph with $O(\log n)$-bit integer weights.

\begin{theorem}
Given a set $P$ of $n$ points and a set $R$ of $n$ weighted rectangles in $\mathbb{R}^2$, any algorithm for computing the minimum weight size-6 set cover requires $\Omega(n^{2-\delta})$ time for any constant $\delta>0$, assuming the Weighted 4-Clique Hypothesis.
\end{theorem}
\begin{proof}
We will reduce the minimum-weight 4-clique problem on a graph with $n$ vertices and $m$ edges ($m\in [n,n^2]$) to the weighted size-6 set cover problem for $O(m)$ points and rectangles in $\R^2$.
Thus, if there is an $O(m^{2-\delta})$-time algorithm for the latter problem, there would be an algorithm for the former problem with running time $O(m^{2-\delta})\le O(n^{4-2\delta})$, refuting the hypothesis.

Let $G=(V,E)$ be the given weighted graph with $n$ vertices and $m$ edges.
Without loss of generality, assume that all edge weights are in $[0,0.1]$, and that
$V\subset [0,0.1]$.
Assume that $0\in V$ and $0.1\in V$.

\subparagraph{The reduction.} For each vertex $t\in V$, create 8 points
\[\begin{array}{rrrr}
(0,&2-t) &\qquad\qquad\mbox{(type 1)}\\
(t,&0) &\mbox{(type 2)}\\
(2,&t) &\mbox{(type 3)}\\
(2-t,&2) &\mbox{(type 4)}\\
(1,&t) &\mbox{(type 5)}\\
(1,&2-t) &\mbox{(type 6)}\\
(t,&1) &\mbox{(type 7)}\\
(2-t,&1) &\mbox{(type 8)}
\end{array}\]

\IGNORE{
\[ \begin{array}{rrrr}
(t,&0),\\
(2,&t),\\
(2-t,&2),\\
(0,&2-t),\\
(1,&t),\\
(1,&2-t),\\
(t,&1),\\
(2-t,&1).
\end{array}
\]
\begin{align*}
(0,t),\\
(2+t,0),\\
(2.1,1+t),\\
(t,1.1),\\
(1,t),\\
(1,1+t),\\
(0,2+t),\\
(3+t,0).
\end{align*}
}

Create the following rectangles in $\mathbb{R}^2$:
\[\begin{array}{rcccccr}
&\forall x_1x_2'\in E:\ R_{x_1x_2'}^{(1)}&=&[0,x_2')&\times& [0,2-x_1] &\ \ \ \mbox{(type 1)}\\
&\forall x_2x_3'\in E:\ R_{x_2x_3'}^{(2)}&=&[x_2,2]&\times& [0,x_3') &\mbox{(type 2)}\\
&\forall x_3x_4'\in E:\ R_{x_3x_4'}^{(3)}&=&(2-x_4',2]&\times& [x_3,2] &\mbox{(type 3)}\\
&\forall x_4x_1'\in E:\ R_{x_4x_1'}^{(4)}&=&[0,2-x_4]&\times& (2-x_1',2] &\mbox{(type 4)}\\
&\forall x_1''x_3''\in E:\ R_{x_1''x_3''}^{(5)}&=&[1,1.1]&\times& [x_3'',2-x_1''] &\mbox{(type 5)}\\
&\forall x_2''x_4''\in E:\ R_{x_2''x_4''}^{(6)}&=&[x_2'',2-x_4'']&\times& [1,1.1] &\mbox{(type 6)}
\end{array}\]
\IGNORE{
\[ \begin{array}{rcccc}
R_{x_1x_2'}^{(1)}&=&[0,x_2')&\times& [0,2-x_1],\\
R_{x_2x_3'}^{(2)}&=&[x_2,2]&\times& [0,x_3'),\\
R_{x_3x_4'}^{(3)}&=&(2-x_4',2]&\times& [x_3,2],\\
R_{x_4x_1'}^{(4)}&=&[0,2-x_4]&\times& (2-x_1',2],\\
R_{x_1''x_3''}^{(5)}&=&[1,1.1]&\times& [x_3'',2-x_1''],\\
R_{x_2''x_4''}^{(6)}&=&[x_2'',2-x_4'']&\times& [1,1.1].
\end{array}
\]
\begin{align*}
R_{x_1x_2'}&=[0,x_1]\times [x_2',2+x_1],\\
R_{x_2x_3'}&=[0,2+x_3']\times [0,x_2],\\
R_{x_3x_4'}&=[2+x_3,3+x_3]\times [0,1+x_4'],\\
R_{x_4x_1'}&=[x_1',3]\times [1+x_4,1.5],\\
R_{x_2''x_4''}&=[1,1.5]\times [x_2'',1+x_4''],\\
R_{x_1''x_3''}&=[0,3+x_3'']\times [2+x_1'',3],\\
R_{x_3'''}&=[3+x_3''',4]\times [0,3].
\end{align*}
}
The weight of each rectangle is set to be the number of points it covers plus the weight of the edge it represents. The total number of points and rectangles is $O(n)$ and $O(m)$ respectively. The reduction is illustrated in Figure~\ref{fig:6_set_cover_weighted_2D}. 

\begin{figure}\centering
    \includegraphics[width=.35\textwidth]{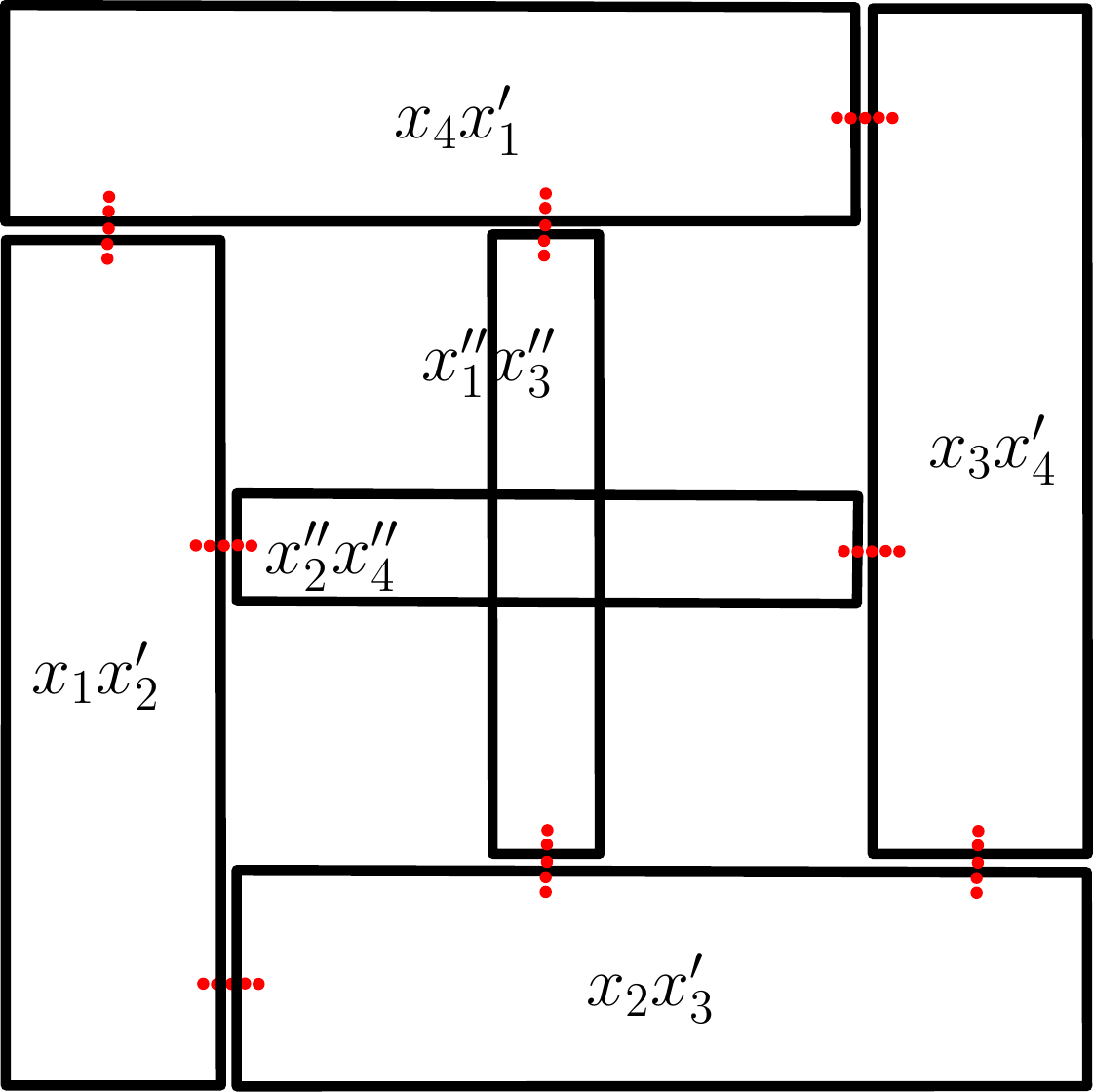}\\
    \caption{Reduction from the minimum-weight 4-clique problem to weighted size-6 set cover for rectangles in $\mathbb{R}^2$.}
    \label{fig:6_set_cover_weighted_2D}
\end{figure}

\subparagraph{Correctness.}
We prove that the minimum-weight 4-clique in $G$ has weight $w$ (where $w\in [0,0.6]$) 
iff the optimal weighted size-6 set cover has weight $8n+w$.

Any feasible solution (if exists) must use a rectangle of each type, since the point $(0,1.9)$ of type 1 (resp.\ the point $(0.1,0)$ of type 2, the point $(2,0.1)$ of type 3, the point $(1.9,2)$ of type 4, the point $(1,0.1)$ of type 5, and the point $(0.1,1)$ of type 7) can only be covered by a rectangle of type 1 (resp.\ 2, 3, 4, 5 and 6). So the six rectangles in a feasible solution must be of the form $R^{(1)}_{x_1x_2'}$, $R^{(2)}_{x_2x_3'}$, $R^{(3)}_{x_3x_4'}$, $R^{(4)}_{x_4x_1'}$, $R^{(5)}_{x_1''x_3''}$ and $R^{(6)}_{x_2''x_4''}$ for some $x_1x_2'$, $x_2x_3'$, $x_3x_4'$, $x_4x_1'$, $x_1''x_3''$, $x_2''x_4''\in E$.

If $x_1>x_1'$, some point (of type 1) would be uncovered; on the other hand, if $x_1<x_1'$, some point (of type 1) would be covered twice, and the total weight would then be at least $8n+1$. Thus, $x_1=x_1'$. Similarly, $x_1'=x_1''$, $x_2=x_2'=x_2''$, $x_3=x_3'=x_3''$, and $x_4=x_4'=x_4''$. So, $x_1x_2x_3x_4$ forms a 4-clique in $G$. We conclude that the minimum-weight solution $R^{(1)}_{x_1x_2}$, $R^{(2)}_{x_2x_3}$, $R^{(3)}_{x_3x_4}$, $R^{(4)}_{x_4x_1}$, $R^{(5)}_{x_1x_3}$ and $R^{(6)}_{x_2x_4}$ correspond to the minimum-weight 4-clique $x_1x_2x_3x_4$ in $G$.
\end{proof}

\section{Conditional Lower Bound for Euclidean Discrete 2-Center}\label{sec:d2c_lb}

In this section, we prove our conditional lower bound for the Euclidean discrete 2-center problem in a sufficiently large constant dimension.
The general structure of our proof is inspired by
Bringmann et al.'s recent conditional hardness proof~\cite{bringmann2022towards} for the problem of computing diameter of box intersection graphs in $\R^{12}$, specifically, testing whether the diameter is more than~2. 
(Despite the apparent dissimilarities of the two problems, what led us to initially suspect that the ideas there might be useful is that both problems are concerned with paths of length~2 in certain geometrically defined graphs, and both problems have a similar ``quantifier structure'', after unpacking the problem definitions.)  Extra ideas are needed, as we are dealing with the Euclidean metric instead of boxes; we end up needing an extra dimension, with carefully tuned coordinate values, to make the proof work.



\IGNORE{
Our contribution is in realizing that their proof approach can actually be applied to Euclidean discrete 2-center, even though
the two problems appear different.  (The fact
that both problems are about paths of length~2 in certain geometrically 
defined graphs is what led us to initially suspect there could be a connection; still, one problem is about $L_\infty$, and the other is
about Euclidean.)  The generalization of the proof to Euclidean discrete $k$-center for larger $k$ follows the same approach, but with further technical ideas. 
}

\begin{theorem}\label{thm:lb:d2c}
	For any constant $\delta > 0$, there is no $O(n^{2-\delta})$-time algorithm for Euclidean discrete 2-center in $\mathbb{R}^{13}$, assuming the Hyperclique Hypothesis.
\end{theorem}
\begin{proof}
    We will reduce the problem of detecting a 6-hyperclique in a 3-uniform hypergraph with $n$ vertices, to the Euclidean discrete 2-center problem on $N=O(n^3)$ points in $\R^{13}$.  Thus, if there is an $O(N^{2-\delta})$-time algorithm for the latter problem, there would be $O(n^{6-3\delta})$-time algorithm for the former problem, refuting the Hyperclique Hypothesis.
    
    Let $G=(V,E)$ be the given 3-uniform hypergraph.  By a standard color-coding technique~\cite{AlonYZ95}, we may assume that $G$ is $6$-partite, i.e., 
    $V$ is partitioned into 6 parts $V_1,\ldots,V_6$, and each edge in $E$ consists of 3 vertices from 3 different parts.
    The goal is to decide the existence of a 6-hyperclique, i.e., $(x_1,\dots,x_6)\in V_1\times \cdots \times V_6$ such that $\{x_i,x_j,x_k\}\in E$ for all distinct $i,j,k\in \{1,\ldots,6\}$.

    Without loss of generality, assume that $V\subset [0,1]$, i.e., vertices are labelled by numbers that are rescaled to lie in $[0,1]$.
    Let $f,g: [0,1]\rightarrow [0,1]$ be some injective functions
    satisfying $f(x)^2+g(x)^2=1$.
    For example, we can take $f(x)=\cos x$ and $g(x)=\sin x$;
    or alternatively, avoiding trigonometric functions,
    $f(x)=x$ and $g(x)=\sqrt{1-x^2}$;
    or avoiding irrational functions altogether,
    $f(x) = 2x/(x^2+1)$ and
    $g(x) = (x^2-1)/(x^2+1)$.  (With the last two options, by rounding to $O(\log n)$ bits of precision, it is straightforward to make our reduction work in the standard integer word RAM model.)

    \subparagraph{The reduction.}  We construct the following set $S$ of $O(n^3)$ points in $\R^{13}$:

    \begin{enumerate}
    \item  For each $(x_1,x_2,x_3)\in V_1\times V_2 \times V_3$ with $\{x_1,x_2,x_3\}\in E$, create the point 
	\[	p_{x_1x_2x_3}\ =\
	\left(f(x_1), g(x_1), f(x_2), g(x_2),f(x_3), g(x_3), 0 ,0,0,0,0,0,1\right).
	\]

    \item Similarly, for each $(x_4,x_5,x_6)\in V_4\times V_5 \times V_6$ such that $\{x_4,x_5,x_6\}\in E$, create the point
	\[	
	q_{x_4x_5x_6}\ =\ \left(0,0,0,0,0,0,f(x_4), g(x_4), f(x_5), g(x_5),f(x_6), g(x_6),-1\right). 
	\]

    \item 	For each 
    $(v_i,v_j,v_k)\in V_i\times V_j\times V_k$ with distinct $i,j,k$ such that $\{v_i,v_j,v_k\}\not\in E$,
    $\{i,j,k\}\neq \{1,2,3\}$, and $\{i,j,k\}\neq\{4,5,6\}$, create a point  
    $z_{v_iv_jv_k}$: the coordinates in dimensions $2i - 1,2i$ are $-f(v_i),-g(v_i)$, and similarly the coordinates in dimensions $2j - 1,2j,2k - 1,2k$ are $-f(v_j), -g(v_j),-f(v_k), -g(v_k)$, respectively; the 13-th coordinate is 
	\[
	\phi_{ijk}\ =\ \abs{\{1,2,3\}\cap\{i,j,k\}}-1.5\ \in\ \left\{-0.5,0.5\right\};
	\] 
	and all other coordinates are 0. For example, if $i=1,j=2,k=4$,
    \[
    z_{v_1v_2v_4}\ =\ \left(-f(v_1),-g(v_1),-f(v_2),-g(v_2),0,0,-f(v_4),-g(v_4),0,0,0,0,0.5\right).
    \]

    \item     Finally, add two auxiliary points $s_{\pm}=\left(0,\dots,0,\pm3.5\right)$.
    \end{enumerate}

    We solve the discrete 2-center problem on the above point set $S$,
    and return true iff the minimum radius is strictly less than $\sqrt{10.25}$.

    \subparagraph{Correctness.}
    Suppose there exists a 6-hyperclique $(x_1,\ldots,x_6)\in V_1\times\cdots\times V_6$ in $G$.  We claim that every point of $S$ has distance strictly less than $\sqrt{10.25}$ from $p_{x_1x_2x_3}$ or $q_{x_4x_5x_6}$.  Thus, $S$ can be covered by 2 balls centered at $p_{x_1x_2x_3}$ and $q_{x_4x_5x_6}$ with radius less than $\sqrt{10.25}$.  To verify the claim,
    consider a point $z_{v_1v_2v_4}\in S$
    for a triple $(v_1,v_2,v_4)\in V_1\times V_2\times V_4$ with $\{v_1,v_2,v_4\}\not\in E$.
    Observe that the distance between the points $(f(v_\ell),g(v_\ell))$ and $(-f(x_\ell),-g(x_\ell))$ in $\R^2$
    is at most 2, with equality iff $v_\ell=x_\ell$. On the other hand, the distance between $(f(v_\ell),g(v_\ell))$ and $(0,0)$ is 1, and the distance between $(0,0)$ and $(-f(x_\ell),-g(x_\ell))$ is 1.  Thus,
    \[ \|z_{v_1v_2v_4}-p_{x_1x_2x_3}\|^2\ \le\ 2^2 + 2^2+ 1 + 1 + 0 + 0 + (0.5-1)^2\ \le\ 10.25, \]
    with equality iff $v_1=x_1$ and $v_2=x_2$.
    Furthermore, 
    \[ \|z_{v_1v_2v_4}-q_{x_4x_5x_6}\|^2\ \le\ 1+1+ 0 + 2^2+1+1+ (0.5+1)^2\ \le\ 10.25, \]
    with equality iff $v_4=x_4$.
    Since $\{x_1,x_2,x_4\}\in E$,
    we cannot have simultaneously $v_1=x_1$, $v_2=x_2$, and $v_4=x_4$.  So, 
    $z_{v_1v_2v_4}$
    has distance strictly less than $\sqrt{10.25}$ from $p_{x_1x_2x_3}$ or $q_{x_4x_5x_6}$. 
    Similarly, the same holds for $z_{v_iv_jv_k}\in S$ for all other choices of $i,j,k$.
    Points $p_{x_1'x_2'x_3'}\in S$ have
    distance at most $\sqrt{2+2+2+0+0+0+0}<\sqrt{10.25}$
    from $p_{x_1x_2x_3}$, and similarly, points $q_{x_4'x_5'x_6'}\in S$ have distance less than $\sqrt{10.25}$ from $q_{x_4x_5x_6}$.
    Finally, the auxiliary point $s_{+}$ has distance at most
    $\sqrt{1+1+1+0+0+0+2.5^2}<\sqrt{10.25}$ from $p_{x_1x_2x_3}$, and similarly the point $s_-$ has distance less than $\sqrt{10.25}$ from $q_{x_4x_5x_6}$.

    On the reverse direction, suppose that the minimum radius for the discrete 2-center problem on $S$ is strictly less than $\sqrt{10.25}$.
    Note that the distance between $s_+$ and $z_{v_iv_jv_k}$ is at least $\sqrt{1+1+1+0+0+0+3^2}>\sqrt{10.25}$,
    and the distance between $s_+$ and
    $q_{x_4x_5x_6}$ is at least
    $\sqrt{0+0+0+1+1+1+4.5^2}>\sqrt{10.25}$.
    Thus, in order to cover $s_+$, one of the two centers must be equal to $p_{x_1x_2x_3}$
    for some $\{x_1,x_2,x_3\}\in E$.
    Similarly, in order to cover $s_-$,
    the other center must be equal to $q_{x_4x_5x_6}$ for some $\{x_4,x_5,x_6\}\in E$.  Then for every $(v_1,v_2,v_4)\in V_1\times V_2\times V_4$ with $\{v_1,v_2,v_4\}\not\in E$, the point $z_{v_1v_2v_4}$ has distance strictly less than $\sqrt{10.25}$ from $p_{x_1x_2x_3}$ or $q_{x_4x_5x_6}$.  By the above argument,
    we cannot have $v_1=x_1$ and $v_2=x_2$ and $v_4=x_4$.  It follows that $\{x_1,x_2,x_4\}\in E$.  Similarly, $\{x_i,x_j,x_k\}\in E$ for all other choices of $i,j,k$.  We conclude that $\{x_1,\ldots,x_6\}$ is a 6-hyperclique.
\end{proof}

From the same proof  (after rescaling), we immediately get a near-quadratic conditional lower bound for unweighted size-2 geometric set cover for unit balls in $\R^{13}$.
In Sec.~\ref{app:dkc},
we extend the proof to Euclidean discrete $k$-center for larger constant $k$, with more technical effort and more delicate handling of the extra dimensions.  This is interesting: discrete $k$-center seems even farther away from graph diameter, but in a way, our proof shows that discrete $k$-center is a better problem to illustrate the full power of Bringmann et al.'s technique~\cite{bringmann2022towards}.

In Sec.~\ref{sec:max:cov}, 
we also adapt the approach to prove a conditional lower bound for size-2 maximum coverage for boxes. 
The proof uses a different way to enforce conditions like $\{x_1,x_2,x_4\}\in E$, via an interesting counting argument.

\section{Conditional Lower Bounds for Euclidean Discrete $k$-Center for $k\ge 3$}
\label{app:dkc}

\newcommand{\kk}{\kappa}

Continuing Section~\ref{sec:d2c_lb},
we extend 
our conditional lower bound proof for Euclidean discrete 2-center to Euclidean discrete $k$-center for larger $k$. 
This requires some extra technical ideas, in particular, using $k$ extra dimensions (instead of one) and with more carefully tuned coordinate values.
In the proof below, it is more convenient to use $k$ for another index, and so we will rename $k$ as $\kk$.

\begin{theorem}\label{thm:lb:dkc}
	For any constant $\delta > 0$, there is no $O(n^{\kk-\delta})$-time algorithm for Euclidean discrete $\kk$-center in $\mathbb{R}^{7\kk}$, assuming the Hyperclique Hypothesis.
\end{theorem}
\begin{proof}
    We will reduce the problem of detecting a $3\kk$-hyperclique in a 3-uniform hypergraph with $n$ vertices, to the Euclidean discrete $\kk$-center problem on $N=O(n^3)$ points in $\R^{7\kk}$.  Thus, if there is an $O(N^{\kk-\delta})$-time algorithm for the latter problem, there would be $O(n^{3\kk-3\delta})$-time algorithm for the former problem, refuting the Hyperclique Hypothesis.
    
    Let $G=(V,E)$ be the given 3-uniform hypergraph.  By a standard color-coding technique~\cite{AlonYZ95}, we may assume that $G$ is $(3\kk)$-partite, i.e., 
    $V$ is partitioned into $3\kk$ parts $V_1,\ldots,V_{3\kk}$, and each edge in $E$ consists of 3 vertices from 3 different parts.
    The goal is to decide the existence of a $3\kk$-hyperclique, i.e., $(x_1,\dots,x_{3\kk})\in V_1\times \cdots \times V_{3\kk}$ such that $\{x_i,x_j,x_k\}\in E$ for all distinct $i,j,k\in \{1,\ldots,3\kk\}$.

    As before, assume that $V\subset [0,1]$, and 
    let $f,g: [0,1]\rightarrow [0,1]$ be some injective functions
    satisfying $f(x)^2+g(x)^2=1$.

    \subparagraph{The reduction.}  We construct the following set $S$ of $O(n^3)$ points in $\R^{7\kk}$:

For each $t\in\{1,\ldots,\kk\}$ and $(x_{3t-2},x_{3t-1},x_{3t})\in V_{3t-2}\times V_{3t-1} \times V_{3t}$ with $\{x_{3t-2},x_{3t-1},x_{3t}\}\in E$, create the point $p^{(t)}_{x_{3t-2}x_{3t-1}x_{3t}}=(p_1,\ldots,p_{7\kk})$ where
\[
\begin{array}{ll}
	p_{6t-5}=f(x_{3t-2}),\qquad\qquad &
	p_{6t-4}=g(x_{3t-2}),\\
	p_{6t-3}=f(x_{3t-1}), &
	p_{6t-2}=g(x_{3t-1}),\\
	p_{6t-1}=f(x_{3t}), &
	p_{6t}=g(x_{3t}),\\
	p_{6\kk+t}=2, &
\end{array}
\]
and all other coordinates are 0. 

For each $(v_i,v_j,v_k)\in V_i\times V_j\times V_k$ with distinct $i,j,k$ such that $\{v_i,v_j,v_k\} \not\in E$ and $\alpha=\ceil{i/3},\beta=\ceil{j/3},\gamma=\ceil{k/3}$ are not identical, create a point $z^{(ijk)}_{v_iv_jv_k}=(z_1,\ldots,z_{7\kk})$ where
	\[ \begin{array}{ll}
	z_{2i-1}=-f(v_i),\qquad\qquad &
	z_{2i}=-g(v_i),\\
	z_{2j-1}=-f(v_j), &
	z_{2j}=-g(v_j),\\
	z_{2k-1}=-f(v_k), &
	z_{2k}=-g(v_k).
	\end{array} \]
	In addition, if $\alpha,\beta,\gamma$ are distinct, we let
	\[
	z_{6\kk+\alpha}=z_{6\kk+\beta}=z_{6\kk+\gamma}=2;
	\]
	otherwise, assuming $\alpha=\beta$ without loss of generality, we let
	\[
	z_{6\kk+\alpha}=\mu,\qquad\qquad
	z_{6\kk+\gamma}=\nu,
	\]
        where $\mu=(5+\sqrt{39})/4\approx 2.8112$ and $\nu=(3+\sqrt{39})/4\approx 2.3112$ satisfy the equations
        $\mu^2+(\nu-2)^2=8$ and $(\mu-2)^2+\nu^2=6$.
	All other coordinates are 0.

    Finally, for each $t\in\{1,\ldots,\kk\}$, we define the auxiliary point $s^{(t)}=(s_1,\ldots,s_{7\kk})$ where $s_{6\kk+t}=5.6$ and all other coordinates are 0.

        We solve the discrete $\kk$-center problem on the above point set $S$,
    and return true iff the minimum radius is strictly less than $4$.

    \subparagraph{Correctness.}
    Suppose there exists a $3\kk$-hyperclique $(x_1,\ldots,x_\kk)\in V_1\times\cdots\times V_\kk$ in $G$.  We claim that every point of $S$ has distance strictly less than $4$ from $p^{(t)}_{x_{3t-2}x_{3t-1}x_{3t}}$ for some $t\in\{1,\ldots,\kk\}$.  Thus, $S$ can be covered by $\kk$ balls centered at $p^{(t)}_{x_{3t-2}x_{3t-1}x_{3t}}$ for  $t\in\{1,\ldots,\kk\}$ with radius less than $4$.  To verify the claim,
    consider a point $z^{(ijk)}_{v_iv_jv_k}\in S$
    for a triple $(v_i,v_j,v_k)\in V_i\times V_j\times V_k$ with $\{v_i,v_j,v_k\}\not\in E$.  Let $\alpha=\ceil{i/3},\beta=\ceil{j/3},\gamma=\ceil{k/3}$.  We consider the following cases, as depicted in Figure~\ref{fig:dkc} (all remaining cases are symmetric):

\begin{figure}
\includegraphics[scale=0.74]{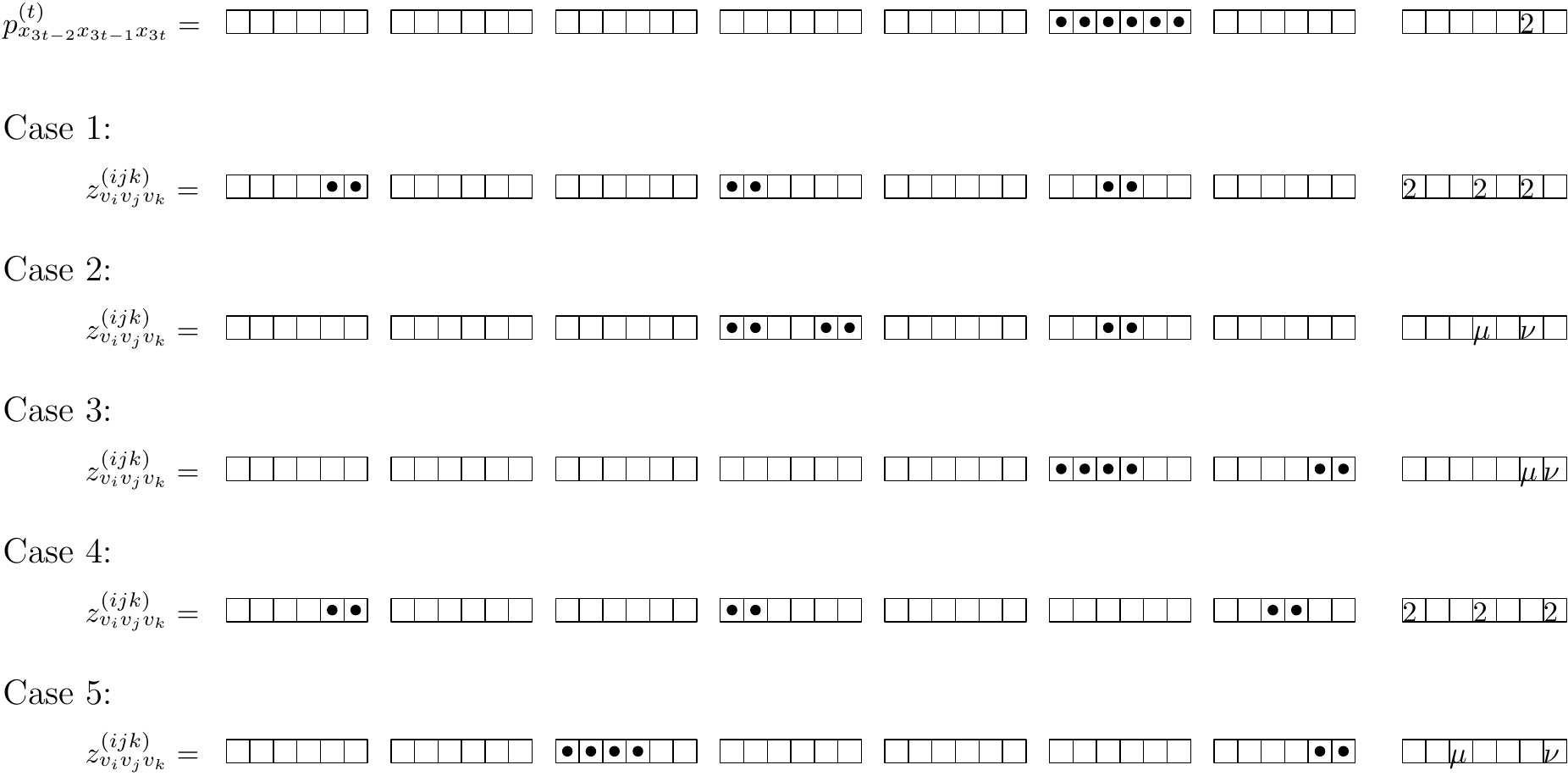}
\vspace{2ex}
\caption{The nonzero coordinates of $p^{(t)}_{x_{3t-2}x_{3t-1}x_{3t}}$ and $z^{(ijk)}_{v_iv_jv_k}$.}\label{fig:dkc}
\end{figure}

    \begin{itemize}
    \item Case 1: If $\{3t-2,3t-1,3t\}\cap \{i,j,k\}=\{k\}$ and $\alpha,\beta,\gamma$ are distinct,
    then
    \[ \|z^{(ijk)}_{v_iv_jv_k}-p^{(t)}_{x_{3t-2}x_{3t-1}x_{3t}}\|^2\ \le\ 2^2 + 1 + 1 + 1 + 1 + (2-2)^2 + (2-0)^2 + (2-0)^2 \ \le\ 16, \]
    with equality iff $v_k=x_k$.
    
    \item Case 2: If $\{3t-2,3t-1,3t\}\cap \{i,j,k\}=\{k\}$ and $\alpha=\beta$,
    then
    \[ \|z^{(ijk)}_{v_iv_jv_k}-p^{(t)}_{x_{3t-2}x_{3t-1}x_{3t}}\|^2\ \le\ 2^2 + 1 + 1 + 1 + 1 + (\mu-0)^2 + (\nu-2)^2 \ \le\ 16, \]
    with equality iff $v_k=x_k$.

    \item Case 3: If $\{3t-2,3t-1,3t\}\cap \{i,j,k\}=\{i,j\}$ and $\alpha=\beta$,
    then
    \[ \|z^{(ijk)}_{v_iv_jv_k}-p^{(t)}_{x_{3t-2}x_{3t-1}x_{3t}}\|^2\ \le\ 2^2 + 2^2 + 1 + 1 + (\mu-2)^2 + (\nu-0)^2 \ \le\ 16, \]  
    with equality iff $v_i=x_i$ and $v_j=x_j$.

    \item Case 4: If $\{3t-2,3t-1,3t\}\cap \{i,j,k\}=\emptyset$ and $\alpha,\beta,\gamma$ are distinct,
    \[ \|z^{(ijk)}_{v_iv_jv_k}-p^{(t)}_{x_{3t-2}x_{3t-1}x_{3t}}\|^2\ =\ 1 + 1 + 1 + 1 + 1 + 1 +  (2-0)^2 + (2-0)^2 + (2-0)^2 + (2-0)^2 \ >\ 16. \]

    \item Case 5: If $\{3t-2,3t-1,3t\}\cap \{i,j,k\}=\emptyset$ and $\alpha=\beta$,
    \[ \|z^{(ijk)}_{v_iv_jv_k}-p^{(t)}_{x_{3t-2}x_{3t-1}x_{3t}}\|^2\ =\ 1 + 1 + 1 + 1 + 1 + 1 +  (\mu-0)^2 + (\nu-0)^2 + (2-0)^2 \ >\ 16. \]
    \end{itemize}

    \noindent
    Since $\{x_i,x_j,x_k\}\in E$,
    we cannot have simultaneously $v_i=x_i$, $v_j=x_j$, and $v_k=x_k$.  So, 
    $z_{v_iv_jv_k}$
    has distance strictly less than $4$ from some $p^{(t)}_{x_{3t-2}x_{3t-1}x_{3t}}$. 
    Points $p^{(t)}_{x_{3t-2}'x_{3t-1}'x_{3t}'}\in S$ have
    distance at most $\sqrt{2^2+2^2+2^2+(2-2)^2}<4$
    from $p^{(t)}_{x_{3t-2}x_{3t-1}x_{3t}}$.
    Finally, the auxiliary point $s^{(t)}$ has distance at most
    $\sqrt{1+1+1+(5.6-2)^2}<4$ from $p^{(t)}_{x_{3t-2}x_{3t-1}x_{3t}}$.

    On the reverse direction, suppose that the minimum radius for the discrete $3\kk$-center problem on $S$ is strictly less than $4$.
    Note that the distance between $s^{(t)}$ and $z^{(ijk)}_{v_iv_jv_k}$ is at least  $\min\{ \sqrt{1+1+1 + (5.6-2)^2+(0-2)^2+(0-2)^2},\ \sqrt{1+1+1+(5.6-\mu)^2 + \nu^2}\}>4$,  
    and for any $t'\neq t$, the distance between $s^{(t)}$ and $p^{(t')}_{x_{3t'-2}x_{3t'-1}x_{3t'}}$ is at least\\    $\sqrt{1+1+1+ (5.6-0)^2+(0-2)^2}>4$.
    Thus, in order to cover the $\kk$ points $s^{(t)}$, for each $t\in\{1,\ldots,\kk\}$, there must be one center equal to $p^{(t)}_{x_{3t-2}x_{3t-1}x_{3t}}$
    for some $\{x_{3t-2},x_{3t-1},x_{3t}\}\in E$.
    Then for every $(v_i,v_j,v_k)\in V_i\times V_j\times V_k$ with $\{v_i,v_j,v_k\}\not\in E$, the point $z^{(ijk)}_{v_iv_jv_k}$ has distance less than $4$ from one of these points $p^{(t)}_{x_{3t-2}x_{3t-1}x_{3t}}$.  By the above case analysis,
    we cannot have $v_i=x_i$ and $v_j=x_j$ and $v_k=x_k$.  It follows that $\{x_i,x_j,x_k\}\in E$.  We conclude that $\{x_1,\ldots,x_{3\kk}\}$ is a $3\kk$-hyperclique.
\end{proof}

\IGNORE{
\subsection{Geometric hitting set of size 2}
Let $(S,X)$ be a range space, where $S$ is a set of points in $\mathbb{R}^d$ and $X$ is a configuration of regions. The geometric hitting set problem for $(S,X)$ is to find a minimum-cardinality subset $T\subseteq S$ such that each region in $X$ contains at least one point in $T$. 
\begin{theorem}
	For all $\epsilon > 0$ there is no $O(n^{2-\epsilon})$ algorithm for finding a geometric hitting set of size 2 in $\mathbb{R}^{13}$, unless the Hyperclique Hypothesis fails.
\end{theorem}
\begin{proof}
	It suffices to reduce discrete 2-center to geometric hitting set of size 2. Given a decision instance $(S,r)$ of discrete 2-center, consider the hitting set instance $(S,\cup_{s\in S}B(s,r))$, where $B(s,r)$ denotes the ball centered at $s$ with radius $r$. If $\{x,y\}$ is a hitting set, then for any $s\in S$, we have $x\in B(s,r)$ or $y\in B(s,r)$, i.e. $s\in B(x,r)$ or $s\in B(y,r)$, which means $\{x,y\}$ are discrete 2-centers of $S$.
\end{proof}

}

From the same proof (after rescaling), we immediately get a near-$n^\kk$ conditional lower bound for size-$\kk$ geometric set cover for unit balls  in $\R^{7\kk}$.

We have not optimized the number of dimensions $7\kappa$, to keep the proof simple.  (With more work, it should be reducible slightly to $6\kappa+O(1)$.)

\section{Conditional Lower Bound for Size-2 Maximum Coverage for Boxes}
\label{sec:max:cov}


Continuing Section~\ref{sec:d2c_lb} and  moving from the Euclidean back to the rectilinear setting,
we modify our conditional lower bound proof for Euclidean 2-center to the size-2 maximum coverage problem for boxes.  In fact, the proof works for orthants (and thus unit hypercubes) in $\R^{12}$.


\begin{theorem}\label{thm:lb:max_coverage}
For any constant $\delta > 0$, given $n$ points and $n$ orthants in $\R^{12}$, there is no $O(n^{2-\delta})$-time algorithm to find two orthants covering the largest number of points, assuming the Hyperclique Hypothesis.
\end{theorem}

\begin{proof}
We will reduce the problem of detecting a 6-hyperclique in a 3-uniform hypergraph with $6n$ vertices, to the max-coverage by 2 orthants problem on $N=O(n^3)$ points in $\R^{12}$.  Thus, if there is an $O(N^{2-\delta})$-time algorithm for the latter problem, there would be $O(n^{6-3\delta})$-time algorithm for the former problem, refuting the Hyperclique Hypothesis.

Let $G=(V,E)$ be the given 3-uniform hypergraph.  By a standard color-coding technique~\cite{AlonYZ95}, we may assume that $G$ is $6$-partite, i.e., 
$V$ is partitioned into 6 parts $V_1,\ldots,V_6$, and each edge in $E$ consists of 3 vertices from 3 different parts. Without loss of generality, assume that each part contains exactly $n$ vertices, by adding a certain number of isolated vertices to each part. The goal is to decide the existence of a 6-hyperclique, i.e., $(x_1,\dots,x_6)\in V_1\times \cdots \times V_6$ such that $\{x_i,x_j,x_k\}\in E$ for all distinct $i,j,k\in \{1,\ldots,6\}$.

Without loss of generality, assume that $V\subset [0,1]$, i.e., vertices are labelled by numbers that are rescaled to lie in $[0,1]$.

Let 
$d_{jk}(x_i)$ denote the number of edges that contain $x_i$ and other two vertices in $V_j$ and $V_k$, and $d_k(x_ix_j)$ denote the number of edges that contain the vertices $x_i,x_j$ and another vertex in $V_k$.

\subparagraph{The reduction.}  We construct the following $O(n^3)$ points and orthants in $\R^{12}$:

\begin{enumerate}
\item
For each $(x_1,x_2,x_3)\in V_1\times V_2 \times V_3$ with $\{x_1,x_2,x_3\}\in E$, create the orthant

\vspace{-2ex}
\begin{eqnarray*}
R_{x_1x_2x_3}&=&(-\infty,x_1]\times (-\infty,-x_1]\times (-\infty,x_2]\times (-\infty,-x_2]\times (-\infty,x_3]\times (-\infty,-x_3]\\&& {} \times \R\times \R\times \R\times \R\times \R\times \R.
\end{eqnarray*}

\item
Similarly, for each $(x_4,x_5,x_6)\in V_4\times V_5 \times V_6$ with $\{x_4,x_5,x_6\}\in E$, create the orthant

\vspace{-2ex}
\begin{eqnarray*}
    R_{x_4x_5x_6}&=&\R\times \R\times \R\times \R\times \R\times \R\times\\&& (-\infty,x_4]\times (-\infty,-x_4]\times (-\infty,x_5]\times (-\infty,-x_5]\times (-\infty,x_6]\times (-\infty,-x_6].
\end{eqnarray*}

\item For each $(v_1,v_2,v_4)\in V_1\times V_2\times V_4$ such that $\{v_1,v_2,v_4\}\not\in E$, create the point 
\[ p_{v_1v_2v_4}\ =\ (v_1,-v_1,v_2,-v_2,-\infty,-\infty,v_4,-v_4,-\infty,-\infty,-\infty,-\infty).  
\]
\item For each $(v_1,v_2)\in V_1\times V_2$,
create $d_4(v_1v_2)$ copies of the point
\[ p_{v_1v_2}\ =\ (v_1,-v_1,v_2,-v_2,-\infty,-\infty,\infty,\infty,\infty,\infty,\infty,\infty). \]
\item For each $v_4\in V_4$,
create $d_{12}(v_4)$ copies of the point
\[ p_{v_4}\ =\ (\infty,\infty,\infty,\infty,\infty,\infty,v_4,-v_4,-\infty,-\infty,-\infty,-\infty).
\]
\item Repeat Steps 3--5 but with $(1,2,4)$
replaced by all other triples $(i,j,k)$ 
such that $i,j\in\{1,2,3\}$ and $k\in\{4,5,6\}$,
or $i,j\in\{4,5,6\}$ and $k\in\{1,2,3\}$, with the points defined analogously.
\end{enumerate}

We solve the size-2 maximum coverage problem on the above set of points and orthants, and return true iff the solution covers at least $18(n^2+n)$ points.

\subparagraph{Correctness.}
Consider two orthants $R_{x_1x_2x_3}$
and $R_{x_4x_5x_6}$ with $x_1\in V_1$, \ldots, $x_6\in V_6$, and $\{x_1,x_2,x_3\},\{x_4,x_5,x_6\}\in E$. The point $p_{v_1v_2v_4}$ is covered iff
($v_1=x_1$ and $v_2=x_2$) or ($v_4=x_4$).
Thus, the number of points $p_{v_1v_2v_4}$ with $\{v_1,v_2,v_4\}\not\in E$ from Step~3 that are covered is 
precisely 

\begin{itemize}
    \item $(n^2-d_{12}(x_4)) + (n-d_4(x_1x_2))$ if $\{x_1,x_2,x_4\}\in E$;
    \item $(n^2-d_{12}(x_4)) + (n-d_4(x_1x_2))-1$ if $\{x_1,x_2,x_4\}\not\in E$;
\end{itemize} 

The point $p_{v_1v_2}$ is covered iff
$v_1=x_1$ and $v_2=x_2$.  Thus, the number of points from Step~4 that are covered is $d_4(x_1x_2)$.
Similarly, the point $p_{v_4}$ is covered iff $v_4=x_4$.  
Thus, the number of points from Step~5 that are covered is $d_{12}(x_4)$.
Therefore, the total number of points from Steps 3--5 that are covered is $n^2+n$ if $\{x_1,x_2,x_4\}\in E$, or $n^2+n-1$ if
$\{x_1,x_2,x_4\}\not\in E$.

We conclude that the total number of points covered over all 18 triples $(i,j,k)$
is at least $18(n^2+n)$ iff $\{x_1,\ldots,x_6\}$
is a 6-hyperclique in $G$.
\end{proof}

\section{Additional Remarks}

\IGNORE{
Tables \ref{tbl1}--\ref{tbl3} provide a  summary of the main previous and new results.

Because our conditional lower bound proofs are simple and accessible without needing much prior background, we think that they would make good examples illustrating fine-grained complexity in computational
geometry.  Generally speaking, there has been considerable development on fine-grained complexity in the broader algorithms community 
over the last decade~\cite{virgisurvey}, but to a lesser extent in computational geometry.
}

For ``combinatorial'' algorithms (i.e., algorithms that avoid fast matrix multiplication or algebraic techniques), some of our conditional lower bounds can be improved.
For example, our near-$n^{4/3}$ lower bounds 
on unweighted size-3 set cover for boxes in $\R^3$ and rectilinear discrete 3-center in $\R^4$ (Theorems \ref{thm:lb:box:3d}--\ref{thm:lb:orthant:4d} and Corollary~\ref{cor:lb:d3c}) can be increased to
near-$n^{3/2}$, under the Combinatorial BMM Hypothesis~\cite{virgisurvey} (which asserts that Boolean matrix multiplication of two $n\times n$ matrices requires $\Omega(n^{3-\delta})$ time for combinatorial algorithms).  These follow from the same reductions, because the Combinatorial BMM Hypothesis is equivalent~\cite{WilliamsW18} to the hypothesis that triangle detection in a graph with $n$ vertices requires $\Omega(n^{3-\delta})$ time for combinatorial algorithms.
Also, our near-quadratic lower bound for Euclidean discrete 2-center in $\R^{13}$
(Theorem~\ref{thm:lb:d2c}) still holds in $\R^{9}$ under the Combinatorial Clique Hypothesis (which asserts that $\ell$-clique detection in a graph with $n$ vertices requires $\Omega(n^{\ell-\delta})$ time for combinatorial algorithms), by reducing from 4-clique in graphs
instead of 6-hyperclique in 3-uniform hypergraphs.  Similarly, the near-quadratic lower bound for Euclidean discrete $k$-center (Theorem~\ref{thm:lb:dkc})
would hold in a somewhat smaller dimension.  Note that although the notion of ``combinatorial'' algorithms is vague, all our algorithms in Sec.~\ref{sec:alg} and Sec.~\ref{app:improve} are combinatorial.
(However, the best combinatorial algorithm we are aware of for Euclidean discrete 2-center in higher dimensions are slower than the best noncombinatorial algorithm, and has time bound near $n^{3-1/O(d)}$ using range searching techniques.)


It is interesting to compare our conditional lower bound proofs with the known hardness proofs by Marx~\cite{Marx05} and 
Chitnis and Saurabh~\cite{ChitnisS22} on fixed-parameter intractibility with
respect to the parameter $k$.  These were also obtained by reduction from $k'$-clique for
some function relating $k$ and $k'$.  Thus, in principle, by setting $k'=3$, they should yield superlinear lower bounds for discrete $k$-center (and related unweighted set cover problems) for some specific constant $k$, but this value of $k$ is probably large (much larger than 3).
On the other hand,
Cabello et al.'s proofs of fixed-parameter intractibility with respect to $d$~\cite{CabelloGKMR11} also produced reductions from $k'$-clique to
Euclidean 2-center and rectilinear 4-center, but the continuous problems are
different from the discrete problems.
(Still, Cabello et al.'s
proofs appeared to use some similar tricks, though ours are simpler.)

One of our initial reasons for studying the fine-grained complexity of
small-size geometric set cover
is to prove hardness of approximation for fast approximation
algorithms for geometric set cover.
For example, Theorem~\ref{thm:lb:box:3d} (with its subsequent remark) implies that no $\OO(n)$-time
approximation algorithm for set cover for fat boxes in $\R^3$ with side lengths $\Theta(1)$ can achieve approximation factor strictly smaller than $4/3$ (otherwise, it would be able to decide whether the optimal size is 3 or at least 4).
For such fat boxes, near-linear-time approximation algorithms with some (large) constant approximation factor follow from known techniques~\cite{AgarwalP20,ChanH20}
(since fat boxes of similar sizes in $\R^3$ have linear union complexity).

\IGNORE{
Our work leaves open a number of questions, for example:

\begin{itemize}
\item Is it possible to make our subquadratic algorithm for rectilinear discrete 3-center in $\R^2$ work in dimension 3 or higher?
\item Is it possible to make our conditional lower bound  proof for rectilinear discrete 3-center in $\R^4$ work in dimension 2 or 3?
\item Is it possible to make our conditional lower bound for Euclidean discrete 2-center in $\R^{13}$ work in dimension 3?
\item Although we have ruled out subquadratic algorithms for Euclidean discrete 2-center in $\R^{13}$, could geometry still help in beating $n^\omega$ time if $\omega>2$?
\end{itemize}
}

\section{Conclusions}

In this paper, we have obtained a plethora of nontrivial new results on a fundamental class of problems in computational geometry related to discrete $k$-center and size-$k$ geometric set cover for small values of $k$.  (See Tables~\ref{tbl1}--\ref{tbl3}.)
In particular, we have a few results where the upper bounds and conditional lower bounds are close:

\begin{itemize}
\item For weighted size-3 set cover for rectangles in $\R^2$, we have given
the first subquadratic $\OO(n^{7/4})$-time algorithm, and an $\Omega(n^{3/2-\delta})$ lower bound
under the APSP Hypothesis.
\item For Euclidean discrete $k$-center (or unweighted size-$k$
set cover for unit balls) in $\R^{O(k)}$,
we have proved an $\Omega(n^{k-\delta})$ lower bound under the Hyperclique Hypothesis,
which is near optimal if $\omega=2$.
\item For size-2 maximum coverage for boxes in a sufficiently large constant dimension,
we have proved an $\Omega(n^{2-\delta})$ lower bound under the Hyperclique Hypothesis,
which is near optimal.
\end{itemize}

For all of our results, we have managed to find simple proofs (each with 1--3 pages).  We view the simplicity and accessibility of our proofs as an asset---they would make good examples illustrating fine-grained complexity techniques in computational
geometry.  Generally speaking, there has been considerable development on fine-grained complexity in the broader algorithms community 
over the last decade~\cite{virgisurvey}, but to a lesser extent in computational geometry.  A broader goal of this paper is to encourage more work at the intersection of these two areas.  We should emphasize that
while our conditional lower bound proofs may appear simple in hindsight, they are not necessarily easy to come up with; for example, see one of our proofs that
require computer-assisted search (Theorem~\ref{thm:lb:orthant:4d}). 

As many versions of the problems studied here still do not have matching upper and
lower bounds, our work raises many interesting open questions. For example:

\begin{itemize}
\item Is it possible to make our subquadratic algorithm for rectilinear discrete 3-center in $\R^2$ work in dimension 3 or higher?
\item Is it possible to make our conditional lower bound proof for rectilinear discrete 3-center in $\R^4$ work in dimension 2 or 3?
\item Is it possible to make our conditional lower bound for Euclidean discrete 2-center in $\R^{13}$ work in dimension 3?
\item Is it possible to make our conditional lower bound for size-2 maximum coverage for boxes in $\R^{12}$ work in dimension 2 or 3?  
\item Although we have ruled out subquadratic algorithms for Euclidean discrete 2-center in $\R^{13}$, could geometry still help in beating $n^\omega$ time if $\omega>2$?
\end{itemize}

We should remark that some of these questions could be quite difficult.  In fine-grained complexity, there are many examples of basic problems that still do not have tight conditional lower bounds (to mention one well-known geometric example, K\"unnemann's recent FOCS'22 paper~\cite{Kun22} has finally obtained a near-optimal conditional lower bound for Klee's measure problem in $\R^3$, but tight lower bounds in dimension 4 and higher are still not known for non-combinatorial algorithms).  Still, we hope that our work would inspire more progress in both upper and lower bounds for this rich class of problems.

\bibliographystyle{plainurl}
\bibliography{b}

\appendix

\end{document}